\documentclass[12pt]{article}

\usepackage[T1]{fontenc}
\usepackage{mathpazo}

\usepackage{amssymb,amsmath,amsthm,bm,mathtools,commath}
\newcommand{\RomNum}[1]{%
  \textup{\uppercase\expandafter{\romannumeral#1}}%
}

\usepackage{pifont}
\newcommand*{\bigs}[1]{\scalebox{1.25}{\ensuremath#1}}

\usepackage{enumitem}
\usepackage[margin=1in]{geometry}
\usepackage{comment}
\usepackage{float}
\usepackage{natbib}
\usepackage{booktabs}
\usepackage{graphicx}
\usepackage{breakcites}
\usepackage{tabularx}

\usepackage{setspace}

\usepackage{multirow}
\usepackage[flushleft]{threeparttable}

\usepackage{makecell}
\usepackage{caption}
\usepackage[dvipsnames]{xcolor}
\usepackage[colorlinks = true,linkcolor = blue]{hyperref}
\definecolor{darkblue}{rgb}{0, 0, 0.5}
\hypersetup{
     colorlinks = true,
     citecolor = Maroon,
     urlcolor = darkblue,
     linkcolor = darkblue
}

\usepackage{subcaption} 

\newtheorem{theorem}{Theorem}
\newtheorem{lemma}{Lemma}

\newtheorem{definition}{Definition}

\newtheorem{example}{Example}
\newtheorem*{remark}{Remark}

\begin{document}

\title{Strength in Numbers: Robust Mechanisms for Public Goods with Many Agents\thanks{
Alphabetical ordering of authors: the two authors contributed equally to this work. We appreciate the Editor Clemens Puppe, an Associate Editor, and two referees whose valuable comments and suggestions have significantly improved this paper. We are grateful to Joel Sobel for his constant support and helpful discussions on this paper. We also thank Songzi Du, Martin Hellwig, and Andrew Postlewaite.}}

\author{Jin Xi\thanks{Department of Economics, University of California, San Diego. Address: 9500 Gilman Dr. La Jolla, CA 92093. Email: \url{x5jin@ucsd.edu}.} \quad \quad Haitian Xie\thanks{Guanghua School of Management, Peking University. Address: 5 Yiheyuan Road, Haidian District, Beijing, China, 100871. Email: \url{xht@gsm.pku.edu.cn}.}}

\date{\textbf{\today}}

\maketitle
\thispagestyle{empty}
\vspace{-2em}

\begin{abstract}
    This study examines the mechanism design problem for public goods provision in a large economy with $n$ independent agents. We propose a class of dominant-strategy incentive compatible and ex-post individually rational mechanisms, which we call the adjusted mean-thresholding (AMT) mechanisms. We show that when the cost of provision grows slower than the $\sqrt{n}$-rate, the AMT mechanisms are both eventually ex-ante budget balanced and asymptotically efficient. When the cost grows faster than the $\sqrt{n}$-rate, in contrast, we show that any incentive compatible, individually rational, and eventually ex-ante budget balanced mechanism must have provision probability converging to zero and hence cannot be asymptotically efficient. The AMT mechanisms have a simple form and are more informationally robust when compared to, for example, the second-best mechanism. This is because the construction of an AMT mechanism depends only on the first moment of the valuation distribution. 

    \vspace{1em}
    \noindent%

\vspace{1em}
\end{abstract}

\section{Introduction}

This study examines the mechanism design problem for public goods in a model of independent private valuations. We take an asymptotic approach to study the \citet{myerson1983efficient}'s impossibility theorem in large economies.\footnote{\citet{myerson1983efficient} establish the impossibility result in the bilateral trade setting, but it extends directly to the public goods problem. We provide a more substantial overview of this problem below.} As a main result, we propose a class of \emph{dominant-strategy incentive compatible}, and \emph{ex-post individually rational} mechanisms that are \emph{eventually ex-ante budget balanced} and \emph{asymptotically efficient} as the population grows. Such mechanisms enjoy simplicity and informational robustness in that the mechanism designer only needs to know the first moment of the valuation distribution. The robustness feature is useful from the designer's perspective and is in line with \citet{wilson1987game}'s doctrine that a mechanism is recommended if it works under a wide range of valuation distributions.

The difficulty of the public goods problem is represented in the tension between budget balance and efficiency. The social planner wants to serve everyone the public good but cannot raise enough money to afford so. This tension is due to the social planner's limited ability in collecting payments for the public good, for which there is a clear economic intuition: there is no way to decline one agent while serving another. It is thus difficult to incentivize the agents to pay for the public good since a single agent's report is unlikely to influence the collective decision. 

In an asymptotic framework where the number of agents goes to infinity, the key to the tension between budget and welfare is the pivotal probability --- the probability that an agent's report is pivotal to the collective decision on the public good. From the efficiency perspective, the pivotal probability needs to decrease fast as the number of agents increases since the social planner wants to provide the public good decisively. From the budget perspective, there needs to be a significant fraction of pivotal agents because they are the ones paying for the public good. The question is whether there exists a rate of decrease for the pivotal probability such that budget balance and efficiency can both be achieved asymptotically as the number of agents increases to infinity.

The first contribution of our paper is to explicitly find such a rate of the pivotal probability and construct simple mechanisms that resolves the tension between budget and welfare.
The mechanism we study is called the \emph{adjusted mean-thresholding} (AMT) mechanism. Here we describe its simplest form. Let there be $n$ agents, and each has a private valuation $V_i \geq 0$ of the public good. The valuations are independent and identically distributed (iid) with mean $\mu$. Consider the following decision rule for providing the public good:\footnote{For any measurable set $E$, we use $\mathbf{1}E$ to denote the indicator function of $E$.}
\begin{equation} \label{eqn:simplest-AMT}
	\mathbf{1}\bigs\{\sum V_i \geq n \mu + \alpha_n \bigs\}.
\end{equation}
Under this decision rule, the public good is provided whenever the sum of valuations exceeds the total expected surplus $n \mu$ adjusted by the term $\alpha_n$. The transfer payments are given by the usual revenue equivalence result such that the incentive constraints are satisfied. The adjustment term $\alpha_n$ is a sequence of constants that goes to negative infinity. It characterizes the rate of decrease for the pivotal probability, which in turn determines the trade-off between budget balance and efficiency. 

We find the following theoretical results. If the cost of providing the public good grows slower than $\sqrt{n}$ (in a way defined more precisely later), then by setting the rate of $ \lvert \alpha_n \rvert$ to be between $\sqrt{n}$ and $\sqrt{n \log n}$, we can make the mechanism described by (\ref{eqn:simplest-AMT}) enjoy two desirable properties: (1) the ex-ante budget is always balanced for $n$ large enough (i.e., eventually ex-ante budget balanced), and (2) the ratio between the achieved welfare and the optimal welfare converges to one (i.e., asymptotically efficient). Moreover, we show that the assumption on the asymptotic growth rate of the cost is inevitable for budget balance because the maximum collectible revenue grows at the $\sqrt{n}$-rate. 

Besides its simplicity, the mechanism described by (\ref{eqn:simplest-AMT}) is informationally robust. The only information required about the distribution of $V_i$ is its first moment $\mu$. In particular, this AMT mechanism is more informationally robust than the \emph{second-best} mechanism studied by \cite{guth1986private}, whose construction relies on the entire distribution function.\footnote{The second-best mechanism maximizes welfare among the class of mechanisms that are incentive compatible, individually rational, and budget balance. We discuss more about the second-best mechanism in Section \ref{literature} and Section \ref{sec:setup}.} The qualitative difference between the mean and the entire valuation distribution can be formalized by using the information theory. Specifically, given an amount of data, the estimation error of the mean is smaller by orders of magnitude than that of the entire distribution.


The second contribution of our paper is in reconciling the positive and negative results in the public good literature regarding the growth rate of the cost. \cite{mailath1990asymmetric} show that it is impossible to resolve the tension between budget and welfare if the cost grows proportionally with $n$. On the other hand, \cite{hellwig2003public} shows that the tension can be resolved if the cost remains fixed (does not grow with $n$). Our results demonstrate that it is the $\sqrt{n}$ growth rate of the cost that draws the line between negative and positive results. This is important in practice when the marginal cost of providing the public good is decreasing in the number of participants. In particular, our results implies that the public good can be efficiently provided under the budget constraint if and only if the marginal cost decreases faster than $1/\sqrt{n}$.

The third contribution of our paper is that we provide a novel approach to calculating the asymptotic budget. By the revenue equivalence result, the expected total payments collected from the mechanism (\ref{eqn:simplest-AMT}) can be represented as a truncated expectation of the sum of virtual valuations:
\begin{align} \label{eqn:simplest-AMT-revenue}
    \mathbb{E} \left[ \sum \psi(V_i) \mathbf{1}\bigs\{\sum (V_i - \mu) \geq \alpha_n \bigs\} \right],
\end{align}
where the virtual valuation is defined as $\psi(V_i) \equiv V_i - (1-F(V_i))/f(V_i)$ with $F$ and $f$ being the cumulative distribution function and the density function, respectively. When $n$ is large, $\sum (V_i - \mu)$ and $\sum \psi(V_i)$ are approximately joint normal by the \emph{central limit theorem}. We can therefore approximate the expectation in (\ref{eqn:simplest-AMT-revenue}) with a normal distribution. When $n$ is finite, however, the joint distribution of $\left(\sum (V_i - \mu),\sum \psi(V_i) \right)$ can deviate from normal distribution and lead to a sizable approximation error that grows with $n$. To bound this approximation error, we use the \emph{multivariate Berry-Esseen theorem}, which specifies the convergence rate of the multivariate central limit theorem.\footnote{The central limit theorem states that the distribution of sum of iid random variables is approximately a normal distribution. The Berry-Esseen theorem further specifies how close the normal approximation is.} More specifically, Berry-Esseen theorem allows us to find the rate of the adjustment term $\alpha_n$ such that the budget of (\ref{eqn:simplest-AMT}) is ensured to balance eventually, even with the presence of approximation error. Such an explicit result on the rate of $\alpha_n$ would be impossible if we only apply the central limit theorem without more accurate convergence results such as Berry-Esseen theorem.\footnote{For example, it would be difficult to derive such a result based on the proof method in \cite{hellwig2003public}. We provide a discussion on this issue in Appendix \ref{sec:comparison-hellwig}.}

The remaining part of this section discusses the literature. Section \ref{sec:setup} introduces the public goods provision problem and the asymptotic approach we take. Section \ref{sec:welfare} proposes the AMT mechanisms and examines the asymptotic behavior of their budget and welfare. Section \ref{sec:conlusion} concludes. Proofs of the results in the main text are listed in Appendix \ref{sec:proofs}.

In a longer version of the paper, we provide several extension results.\footnote{This version can be found at: https://arxiv.org/abs/2101.02423v3.}  These extensions include independent but not necessarily identical valuations, ex-post budget and welfare, non-binary decision environment, non-linear utility functions, and asymptotic profit. Due to space limits, these extensions are not presented in the current paper.


\subsection{Literature Review}\label{literature}

Our study addresses a long-existing challenge in the literature on the mechanism design problem for public goods. The aforementioned tension between budget balance and efficiency is in fact a quadrilemma, as the incentive constraints also play a role in the problem. More concretely, \cite{myerson1983efficient} show that efficiency cannot be achieved simultaneously with the three criteria: budget balance, incentive compatibility, and individual rationality. For example, the VCG mechanism \citep{vickrey1961counterspeculation,clarke1971multipart,groves1973incentives} implements the efficient outcome but always leads to a budget deficit. In fact, our results imply that even the pivot mechanism (originally due to \cite{green1977revelation}, also known as the Clarke mechanism), which is the VCG mechanism with the least budget deficit, can incur a growing budget deficit as the population grows. As another example, the second-best mechanism introduced by \citet{guth1986private} maximizes the welfare under the three criteria but leads to inefficiency.


The next question is whether this tension can be alleviated in the large economy setting, where the number of agents increases to infinity.\footnote{There is another notion of ``large economy'' in the literature that considers a continuum of agents and each agent is negligible compared to the population \citep{bierbrauer2009optimal,bierbrauer2014optimal,bierbrauer2015public,bierbrauer2016robustly}.} As mentioned earlier, the answer depends on the growth rate of cost. \citet{mailath1990asymmetric} let the cost grow proportionally with the population size, and show that for any mechanism satisfying the above three criteria, the probability of public-good provision goes to zero as the number of agents increases. On the other hand, \citet{hellwig2003public} allows the cost to remain fixed and shows that the second-best mechanism is asymptotically efficient. Our results fill the gap between these two papers. We show that the second-best provision probability converges to zero if the cost grows faster than the $\sqrt{n}$-rate. On the other hand, if the cost grows slower than the $\sqrt{n}$-rate, the AMT mechanisms we propose have a provision probability converging to one.\footnote{\cite{hellwig2003public} also identifies the $\sqrt{n}$-rate in the context of Bayesian incentive compatibility. Our results are focused on the dominant-strategy incentive compatibility. One take-away here is that the weaker notion of Bayesian incentive compatibility does not relax this bound. This notion of equivalence between Bayesian incentive compatibility and dominant-strategy incentive compatibility represented by the requirement on the growth rate of cost can be seen as a complement to the theoretical results in \cite{gershkov2013equivalence}.}


Besides the asymptotic approach we take, another case where efficiency can be achieved simultaneously with budget balance, incentive compatibility, and individual rationality is when the valuations are \emph{correlated}, and their joint distribution is \emph{known} to the mechanism designer, as demonstrated by \citet{kosenok2008individually}. Their mechanisms extract the social surplus by taking advantage of the correlation structure among agents' valuations. This result can be seen as an extension of the finding of \citet{cremer1988full} regarding the public goods scenario. Our paper complements their results by examining the case of independent valuations with an unknown distribution.

Our proposed class of mechanisms is novel in that it significantly reduces the information requirements on the valuation distribution as compared with the aforementioned studies. The second-best mechanism requires the full information of valuation distributions. The mechanism by \citet{kosenok2008individually} requires the exact details of the joint distribution of all agents' valuations. The AMT mechanism we propose, in contrast, is appealing in its simplicity and robustness as it requires only one moment from the valuation distribution.\footnote{This moment can come from any increasing and continuous transformation of the valuation. See Section \ref{sec:welfare}.} A moment of a distribution contains much less information than the distribution itself: a moment condition can be satisfied by an infinite number of distributions.\footnote{In Section \ref{ssec:robust}, we provide a more formal discussion of the qualitative difference between a moment and the distribution by using results in the information theory.} The notion of informational robustness under moment restrictions has become popular in recent literature. For example, \citet{carrasco2018optimal} study the maxmin auction under general moment restrictions. \citet{azar2013parametric} and \citet{pinar2017robust} analyze robust auctions when the mean and the variance of the valuations are known.\footnote{These papers examine the model of independent valuations. It is also possible to allow for correlated valuations while maintaining the moment restrictions on the marginal distributions. See \cite{brooks2021maxmin,zhang2021robust}.}

Our study also relates to the literature on the asymptotic optimality results for selling strategies and auction designs. For example, \cite{armstrong1999price} studies tariffs for a multiproduct monopolist firm. The proposed nonlinear tariff is almost optimal and depends only on the mean of the distribution of total surplus. In the auction setting, \citet{swinkels1999asymptotic} shows the asymptotic efficiency of the discriminatory auction when the number of independent bidders grows large.\footnote{Other works regarding the asymptotic efficiency of auctions in the independent private value setting include \citet{swinkels2001efficiency,feldman2016price}.} More recently, in the setting of interdependent valuations, \citet{du2018robust} and \citet{mclean2018very} design robust auctions in which the seller extracts the full surplus asymptotically.\footnote{Other works on asymptotic approximation of the optimal revenue using robust auctions include \citet{segal2003optimal,neeman2003effectiveness,goldberg2006competitive,dhangwatnotai2015revenue}. See also \citet[][Chapter 5]{hartline2016mechanism} for a comprehensive treatment of prior-independent mechanism design.} These results are different from ours due to the intrinsic difference between private and public goods: the auctioneer does not need to worry about the budget constraint and is focused on maximizing revenue rather than welfare.

\section{The Public-Good Provision Problem} \label{sec:setup}

\subsection{Basic setup}

Consider a sequence of communities indexed by $n = 2,3,\cdots$ In the $n$th community, there are $n$ individuals who have to make a joint decision on whether to produce some indivisible and non-excludable public good. The setup of our problem is similar to \cite{kuzmics2017public} and Chapters 3 and 4 in \cite{borgers2015introduction}. The only difference is that in our setting $n$ is not fixed since we focus on conducting an asymptotic analysis where the number of agents in the economy increases toward infinity.


Let $V_{i}$ be the private value of the public good for agent $i$. These valuations are collected into the vector $\bm{V} \equiv (V_1, V_2, \dots, V_n)$. We assume the valuations are \emph{independent and identically distributed} (iid) draws from an unknown distribution $F$, which we call the valuation distribution. We assume the support of $F$ to be $[0,\bar{v}]$,\footnote{In Section \ref{sec:welfare}, we discuss why the support is assumed to be in this form.} where $\bar{v} \in (0,\infty)$ is also unknown to the designer. Assume that $F$ is absolutely continuous with density function $f$. For the relevant expectations to exist, we assume for simplicity that $f$ is bounded away from zero on $[0,\bar{v}]$.\footnote{We can relax the assumptions that $\bar{v}$ is finite and that $f$ is bounded away from zero. We only need to assume that the appropriate moments exist. However, that makes the exposition inconvenient.}

We consider the direct mechanisms: the social planner asks each agent to report their valuations, and then make decisions on (1) whether to provide the public good and (2) transfer payments from the agents.
We use $q^n: \mathbb{R}_+^n \rightarrow \{0,1\}$ to denote a decision rule that assigns each reported valuation vector $\bm{v} \in \mathbb{R}_+^n$ to a collective decision about the public good.\footnote{In this paper, we use upper case letters to denote random variables and lower case letters to denote realizations or reported values.} For each agent $i$, we use $t^n_i:\mathbb{R}_+^n \rightarrow \mathbb{R}$ to describe the transfer that agent $i$ makes based on the valuations. The set of functions $(q^n, \{t_i^n\})$ constitutes a direct mechanism. 

There is a cost associated with the provision of the public good, which is denoted by $c_n > 0$. We explicitly allow the cost to increase with the number of agents. 
The \textit{ex-post budget} of a mechanism is the difference between the sum of the received payments and the incurred cost:
\begin{align} \label{eqn:AMT-budget}
    b^{n}(\bm{V}) & \equiv \sum t^{n}_i(\bm{V}) - c_n q^{n}(\bm{V}), 
\end{align}
The \textit{ex-ante budget} is $\mathbb{E}[b^{n}(\bm{V})]$, the expectation of the ex-post budget. 
We say that a mechanism is \textit{ex-ante budget balanced} if its ex-ante budget is non-negative, that is, $\mathbb{E}[b^{n}(\bm{V})] \geq 0$.\footnote{We use the notion ``budget balance'' as in \citet{kuzmics2017public} and \citet{borgers2015introduction}. It is also termed ``feasible'' or ``subsidy-free'' in the literature. Some literature uses the term ``budget balance'' to refer to the more stringent case where the total payment exactly equals the cost. }

For an individual $i$, her utility is determined by the provision of the public good $q^n$, her private valuation $V_i$, and transfer payment $t_i$. We assume the utility takes the linear form: $V_i q^n - t^n_i$. 
The \textit{ex-post welfare} achieved by a mechanism is the sum of $n$ agents' utilities:
\begin{align}\label{eqn:AMT-welfare}
	w^{n}(\bm{V}) & \equiv  \sum V_i q^{n}(\bm{V}) - \sum t^{n}_i(\bm{V}), 
\end{align}
This is equal to the difference between the sum of the valuations (if the public good is provided) and the total payments.
This welfare is constructed from the utilitarian perspective.\footnote{This definition of social welfare is standard in the literature. From the social planner's perspective, the production cost does not enter into welfare because it does not directly affect agents' utility. The cost only plays a role through the budget constraint.} The \textit{ex-ante welfare} is $\mathbb{E}[w^{n}(\bm{V})]$, the expectation of the ex-post welfare. 

If the social planner observes the private valuations $\bm{V}$, then the public good problem is
\begin{align*}
    \max_{(q^n, \{t_i^n\})} w^{n}(\bm{V}) \text{ subject to } b^{n}(\bm{V}) \geq 0.
\end{align*}
That is, the social planner maximizes the ex-post welfare under the ex-post budget balance constraint. The solution to this problem is to provide the public good whenever the sum of valuations exceeds the cost $c_n$. This provision rule is called the \emph{efficient} (or the \emph{first-best}) decision rule:
\begin{align} \label{eqn:efficient-decision-rule}
	\mathbf{1}\bigs\{ \sum v_i \geq c_n \bigs\},
\end{align}
We use $w^*(\bm{V})$ to denote the \textit{efficient ex-post welfare} achieved under the efficient rule:
\begin{align} \label{eqn:efficient-welfare}
	w^*(\bm{V}) \equiv \bigs( \sum V_i - c_n \bigs) \mathbf{1}\bigs\{ \sum V_i \geq c_n \bigs\}.
\end{align}
The expectation $\mathbb{E} \left[ w^*(\bm{V}) \right] $ is the \textit{efficient ex-ante welfare}. 

In practice, the social planner does not observe the private valuations $\bm{V}$, a problem often referred to as asymmetric information. We need to design the mechanism $(q^n, \{t_i^n\})$ in a way that the agents are willing to participate and report the true valuation. For this purpose, we introduce the following two criteria: dominant-strategy incentive compatibility and ex-post individual rationality.
Let $\bm{v} \in \mathbb{R}_+^n$ be a vector of realized valuations. We follow the convention of using $\bm{v}_{-i}$ to denote the vector $\bm{v}$ excluding the $i$th element. The realized utility of agent $i$ who chooses to report $v'_i$ is
\begin{align} \label{eqn:utility}
	v_i q^n(v'_i,\bm{v}_{-i}) - t^n_i(v'_i,\bm{v}_{-i}).
\end{align}
A mechanism $(q^n, \{t_i^n\})$ is \textit{dominant-strategy incentive compatible} if, for each agent $i$, each set of valuations $\bm{v} \in \mathbb{R}_+^n$, and each report $v_i' \in \mathbb{R}_+$,
\begin{align} \label{eqn:dominant-strategy incentive compatible}
	v_i q^n(v_i,\bm{v}_{-i}) - t^n_i(v_i,\bm{v}_{-i}) \geq v_i q^n(v_i',\bm{v}_{-i}) - t^n_i(v_i',\bm{v}_{-i}).
\end{align}
A mechanism is \textit{ex-post individually rational} if for each agent $i$ and each set of valuations $\bm{v} \in \mathbb{R}_+^n$,
\begin{align} \label{eqn:ex-post individually rational}
	v_i q^n(v_i,\bm{v}_{-i}) - t^n_i(v_i,\bm{v}_{-i}) \geq 0.
\end{align}

By the routine revenue equivalence result in the public-good scenario \citep{borgers2015introduction,kuzmics2017public}, we know that a mechanism is dominant-strategy incentive compatible and ex-post individually rational if and only if the following conditions hold:
\begin{align}
    & q^n \text{ is non-decreasing in every entry}, \label{eqn:q-nondecreasing} \\ 
    & t^n_i(0,\bm{v}_{-i}) \leq 0, \text{ for all } \bm{v}_{-i} \in \mathbb{R}_+^{n-1}, \label{eqn:payment-lowest-type} \\ 
	& t^n_i(\bm{v}) = \hat{v}_i(\bm{v}_{-i})q^n(\bm{v}) + t^n_i(0,\bm{v}_{-i}), \text{ for all } \bm{v} \in \mathbb{R}_+^{n}, \label{eqn:rev-equiv}
\end{align}
where
\begin{align} \label{eqn:pivotal-value}
	\hat{v}_i(\bm{v}_{-i}) \equiv \inf\{ v \geq 0 : q^n(v,\bm{v}_{-i}) =1 \}
\end{align}
denotes the pivotal value of agent $i$ given the other agents' valuations $\bm{v}_{-i}$. An equivalent way to write Equation (\ref{eqn:rev-equiv}) is 
\begin{align*}
    t^n_i(\bm{v}) = v_i q^n(\bm{v}) - \int_0^{v_i} q^n(v,\bm{v}_{-i})dv + t^n_i(0,\bm{v}_{-i}).
\end{align*}
By changing the order of integration, we can write the expected payment from agent $i$ as
\begin{align} 
	\mathbb{E} \left[ t^n_i(\bm{V}) \right] = \mathbb{E} \left[ \psi(V_i)  q^n(\bm{V}) + t^n_i(0,\bm{V}_{-i}) \right], \label{eqn:ex-ante-budget}
\end{align}
where $\psi$ is the \textit{virtual valuation function} defined by
\begin{align} \label{eqn:virtual-valuation}
	\psi: v \mapsto v - \frac{1-F(v)}{f(v)}.
\end{align}
The expectation in (\ref{eqn:ex-ante-budget}) is taken over the valuations $\bm{V}$.

Given the two criteria introduced above, we can specify the public good problem under asymmetric information:
\begin{align*}
    \max_{(q^n, \{t_i^n\})} & \mathbb{E}[w^{n}(\bm{V})] \text{ subject to } \mathbb{E}[b^{n}(\bm{V})] \geq 0, \\
    & \text{ and } (q^n, \{t_i^n\}) \text{ satisfies Conditions (\ref{eqn:q-nondecreasing}) - (\ref{eqn:rev-equiv}).} 
\end{align*}

It is well-known that the efficient mechanism is not the solution to the above problem. In particular, the efficient mechanism cannot simultaneously satisfy dominant-strategy incentive compatibility, ex-post individual rationality, and ex-ante budget balance.\footnote{See, for example, \cite{myerson1983efficient} and \cite{mailath1990asymmetric}.}
The solution to the above problem is referred to as the \textit{second-best mechanism} in the literature. The second-best mechanism suffers from two problems. First, the construction depends on the entire distribution $F$ and hence is not robust. Second, even if we know the distribution, it is difficult to obtain an analytical expression of the mechanism due to the complexity of the optimization problem. In this paper, we do not study the second-best mechanism. Instead, we propose mechanisms that are less informationally demanding and have a simple and explicit analytical expression. Moreover, these mechanisms
achieve welfare comparable to that of the second-best mechanism when $n$ is large.  

\subsection{Asymptotic criteria}

The setup introduced above is standard in the mechanism design literature. Next, we present a novel asymptotic approach for studying the public goods problem. Considering the impossibility result mentioned previously, we propose two optimality criteria, which are respectively the asymptotic relaxation of the ex-ante budget balance condition and efficiency. 

\begin{definition} \label{def:eventually ex-ante budget balanced}
	A sequence of mechanisms $( q^n, \{t^n_i\} )$ (indexed by the number of agents $n$) is \textit{eventually ex-ante budget balanced}, if there is an integer $n_0$ such that
	\begin{align*}
		\mathbb{E} \left[ b^n(\bm{V}) \right]  \geq 0, \forall n>n_0,
	\end{align*}
    where $b^n$ is defined in (\ref{eqn:AMT-budget}).
    That is, the mechanism is ex-ante budget balanced for large enough $n$.
\end{definition}


We examine the ex-ante budget rather than the ex-post one. A rationale for this approach is offered by \cite{Kunimoto2021incentive}. In their Theorem 4, the authors demonstrate that no DSIC and EPIR mechanism can achieve ex-post budget balance, provided that the mechanism adheres to a richness condition. This condition requires that the public good be supplied if all agents, except one, have their highest type—a relatively mild condition in large economies. This reasoning supports the relaxation of the budget balance constraint into eventual EABB.\footnote{The valuations explored in \cite{Kunimoto2021incentive} are discrete. However, this aspect does not impact the results qualitatively, as both discrete and continuous distributions can be addressed using the general measure theory. Specifically, discrete probability measures have density functions with respect to the counting measure, enabling the computation of virtual valuations accordingly.}


\begin{definition} \label{def:asymptotically efficient}
	A sequence of mechanisms $( q^n, \{t^n_i\} )$ is \textit{asymptotically efficient} if as $n \rightarrow \infty$,
	\begin{align} \label{eqn:welfare-ratio}
		\frac{\mathbb{E}[w^*(\bm{V})] - \mathbb{E}\left[ w^n(\bm{V}) \right] }{\mathbb{E}[w^*(\bm{V})]} \rightarrow 0,
	\end{align}
    where $w^n$ and $w^*$ are defined in (\ref{eqn:AMT-welfare}) and (\ref{eqn:efficient-welfare}), respectively.
\end{definition}

The ratio in (\ref{eqn:welfare-ratio}) can be interpreted as the welfare regret ratio of the mechanism designer. The numerator is the difference between the efficient ex-ante welfare and the achieved ex-ante welfare by the mechanism $( q^n, \{t^n_i\} )$, which is often interpreted as the regret of the decision-maker associated with the current valuation distribution. The denominator in (\ref{eqn:welfare-ratio}) is the efficient ex-ante welfare. Therefore, this ratio represents the normalized regret of the mechanism designer. The regret ratio is frequently used in the literature as a criterion to rank probabilistic objects. The axiomatic foundation of such criterion from the maxmin perspective is established by \cite{brafman2000axiomatic}.\footnote{They refer to the ``regret ratio'' in this paper as ``competitive ratio.''} The literature on approximately optimal mechanism design usually aims to find a mechanism that guarantees a fraction of the benchmark welfare or revenue, which can also be considered as the regret ratio criterion. See \citet{roughgarden2019approximately} for a survey of that literature. 

Our goal is to propose dominant-strategy incentive compatible and ex-post individually rational mechanisms that are both eventually ex-ante budget balanced and asymptotically efficient, but involve much less knowledge about the valuation distributions than the second-best mechanism. Moreover, we characterize the growth rate of the budget surplus as well as the convergence rate of the welfare regret ratio for the proposed mechanisms.

\section{Welfare Maximization with Many Agents} \label{sec:welfare}

\subsection{AMT mechanisms}

We want to generalize the mechanism introduced in (\ref{eqn:simplest-AMT}) by considering transformations of the valuations. By allowing for transformations, we can incorporate many important mechanisms into the framework.
Let $h:\mathbb{R}_+ \rightarrow \mathbb{R}_+$ represent a transformation on the space of valuations. We denote $\mu_{h} \equiv \mathbb{E}[h(V_i)]$ as the mean of the transformed value $h(V_i)$ and introduce the following mechanism.

\begin{definition} \label{def:AMT}
	The adjusted mean-thresholding (AMT) mechanism $(q^n,\{t_i^n\})$ with transformation $h$ and adjustment term $\alpha_n$ is defined as
	\begin{alignat*}{2}
		\text{decision rule: } & q^{n}\left( \bm{v} \right) && \equiv \mathbf{1}\bigs\{ \sum h(v_i) \geq n \mu_{h} + \alpha_n \bigs\}, \\
		\text{transfer payment: } & t_i^{n}(\bm{v}) && \equiv  \hat{v}_i^{n}(\bm{v}_{-i})q^{n}(\bm{v}), 1 \leq i \leq n,
	\end{alignat*}
	where $\hat{v}_i^{n}$ is the pivotal value of agent $i$ under the decision rule $q^{n}$ as defined in (\ref{eqn:pivotal-value}). More specifically, the transfer payment $t_i^{n}$ is specified by (\ref{eqn:rev-equiv}) with $t_i^{n}(0,\bm{v_{-i}})$ set to $0$. For simplicity, we suppress the dependence on $h$ and $\alpha_n$ in the notation of the mechanism $(q^{n},\{t_i^{n}\})$. We will explicitly state the transformation $h$ whenever necessary.
\end{definition}

The AMT mechanism is called as such because the mechanism transforms the valuations using $h$ and compares the sum of transformed values, $\sum h(V_i)$, with a threshold. This threshold is the sum of the means, $n \mu_{h}$, adjusted by the term $\alpha_n$. We set the adjustment term to be negative and decreasing towards $-\infty$. The reason $\alpha_n$ is called the adjustment term is that, as shown later, the asymptotic order of $\alpha_n$ is smaller than that of $n \mu_{h}$ for the AMT mechanism to work. That is, the leading term in the threshold is the sum of means $n \mu_{h}$, and $\alpha_n$ is an adjustment that slightly raises the allocation probability of the public good. The reason for introducing the transformation $h$ in Definition \ref{def:AMT} is to include a wide range of mechanisms. Below are some examples that can be seen as special cases of the AMT mechanism with different choices of the transformation $h$ and adjustment term $\alpha_n$.

\begin{example} [Identity Transformation] \label{eg:identity}
	Let $h$ be the identity mapping $id$, i.e., $h(v) = id(v) \equiv v$. Denote $\mu \equiv \mathbb{E}[V_i]$. Then the AMT mechanism with the identity transformation has decision rule
	\begin{align*}
		q^{n}\left( \bm{v} \right) \equiv \mathbf{1}\bigs\{ \sum v_i \geq n \mu + \alpha_n \bigs\}.
	\end{align*}
	This decision rule is the one specified by (\ref{eqn:simplest-AMT}) in the Introduction.
\end{example}

\begin{example} [Pivot Mechanism] \label{eg:pivot}
	The pivot mechanism is the mechanism that is given by the efficient decision rule (\ref{eqn:efficient-decision-rule}) and by the transfer payment defined by (\ref{eqn:rev-equiv}) with $t_i(0,\bm{v_{-i}})$ set to $0$.\footnote{This definition is equivalent to Definition 3.8 in \cite{borgers2015introduction}.} By rewriting the efficient decision rule
	\begin{align*}
		\mathbf{1}\bigs\{ \sum v_i \geq c_n \bigs\} = \mathbf{1}\bigs\{ \sum v_i \geq n \mu + \underbrace{\bigs(c_n - n \mu \bigs)}_{\alpha_n} \bigs\},
	\end{align*}
	we can see that the pivot mechanism is a special case of the AMT mechanism in Example \ref{eg:identity} with a particular choice of the adjustment term $\alpha_n = c_n - n \mu$.
\end{example}

\begin{example} [Virtual Valuation Transformation] \label{eg:vv}
	By setting $h = \psi$, where $\psi$ is defined in (\ref{eqn:virtual-valuation}), we obtain the AMT mechanism with the virtual valuation transformation, which has decision rule:
	\begin{align*}
		q^{n}\left( \bm{v} \right) \equiv \mathbf{1}\bigs\{ \sum \psi(v_i) \geq \alpha_n \bigs\}.
	\end{align*}
	This mechanism chooses to provide the public good if and only if the sum of virtual valuations exceeds the adjustment term $\alpha_n$. Notice that the virtual valuation is always mean zero, that is, $\mu_{\psi}=0$.\footnote{Since $F$ is supported on $[0,\bar{v}]$, we have $\int_{0}^{\bar{v}} (1-F(v)) dv = \mathbb{E}[V_i]$, $\mathbb{E}[\psi(V_i)] = \mathbb{E}[V_i] - \int_{0}^{\bar{v}} (1-F(v)) dv$ = 0. See the end of Section \ref{ssec:hellwig} for a discussion of the support.}
	
	It is important to note that, the construction of this mechanism depends on the entire valuation distribution through $\psi$. However, the ex-ante budget of this mechanism has a simple form, which we examine next.
\end{example}

\subsection{Preliminary analysis with known and regular distribution} \label{ssec:hellwig}

We start the analysis of the AMT mechanisms with the special case of Example \ref{eg:vv}, which requires that the valuation distribution is known to the mechanism designer and satisfies the Myerson regularity condition. This case is used to familiarize the reader with the asymptotic optimality criteria introduced in Definitions \ref{def:eventually ex-ante budget balanced} and \ref{def:asymptotically efficient} and showcase the intuitions of the AMT mechanism in general.

The decision rule $q^{n}$ is non-decreasing when the valuation distribution satisfies the Myerson regularity condition, under which the mechanism $(q^{n},\{t^{n}_i\})$ is dominant-strategy incentive compatible and ex-post individually rational.\footnote{See Equation (\ref{eqn:rev-equiv}) and the discussion before it.} From Equation (\ref{eqn:ex-ante-budget}), we can write the ex-ante budget as
\begin{align} \label{eqn:total-expected-payment-vv}
	\mathbb{E}\bigs[b^{n}(\bm{V})\bigs] = \underbrace{\mathbb{E} \bigs[\sum \psi(V_i) \mathbf{1}\bigs\{ \sum \psi(V_i) \geq \alpha_n \bigs\} \bigs]}_{\text{total expected payment}} - \underbrace{\mathbb{E}\bigs[ c_n \mathbf{1}\bigs\{ \sum \psi(V_i) \geq \alpha_n \bigs\}\bigs]}_{\text{expected cost}}.
\end{align}
The total expected payment is simply the $\alpha_n$-truncated mean of the total virtual valuation $\sum \psi(V_i)$. This simple expression of the total expected payment is why we start the analysis of AMT mechanisms with the transformation $h=\psi$.

We will briefly explain the intuition why the AMT mechanism with $h=\psi$ works in the setting of many agents.\footnote{A similar intuition can be found in the proof of Proposition 3 in \cite{hellwig2003public}. Here, we provide a more detailed discussion regarding the rate of the threshold $\alpha_n$.}
Since the virtual valuations are independent and mean zero, according to the central limit theorem, the total virtual valuation behaves similar to a normal random variable with zero mean. Therefore, the total expected payment is approximately the $\alpha_n$-truncated normal mean. Since values below $\alpha_n$ are truncated, this truncated normal mean is always positive due to the symmetry in the normal distribution. We expect this truncated mean to increase towards infinity if $\alpha_n$ does not move toward negative infinity too quickly. When the truncated mean grows faster than the cost $c_n$, the ex-ante budget is balanced eventually. 

On the other hand, if the cost increases at a slower rate than the sum of expected valuations, the efficient decision rule $\mathbf{1}\{ \sum V_i \geq c_n \}$ nearly always chooses to provide the public good when there are many agents. Therefore, if $\alpha_n \rightarrow -\infty$ fast enough, then the AMT mechanism would almost always make the same decision as the efficient mechanism would, and hence approximate the efficient ex-ante welfare.

The above discussion demonstrates two compelling forces that represent the trade-off between budget balance and welfare maximization. For the budget to balance, we would want the adjustment term to decrease slowly so that the public good is allocated less often and more payments can be collected. For the welfare regret ratio to converge to zero, we would want the adjustment term to decrease rapidly so that the public good is allocated more often and the agents' welfare can be improved. Therefore, the performance of the mechanism crucially depends on the asymptotic behavior of the threshold.

We utilize the Berry-Esseen theorem in probability theory to determine the middle ground of the aforementioned asymptotic trade-off between budget and welfare. The Berry-Esseen theorem is a stronger version of the central limit theorem that characterizes the convergence rate of the sample average towards the normal distribution. It takes into account the approximation error that stems from the difference of the normal distribution and the true distribution, which allows us to obtain a lower bound on the budget.  Our next theorem shows that the AMT mechanism $(q^{n},\{t^{n}_i\})$ can be both eventually ex-ante budget balanced and asymptotically efficient, if (1) the cost does not grow too fast, and (2) the mechanism designer carefully calibrates the adjustment term $\alpha_n$.

\begin{theorem} \label{thm:hellwig} 
	Assume the valuation distribution $F$ is Myerson regular, then the AMT mechanism $(q^{n},\{t^{n}_i\})$ with transformation $h = \psi$ (Example \ref{eg:vv}) is dominant-strategy incentive compatible and ex-post individually rational. The following limiting statements hold true for this mechanism as $n \rightarrow \infty$.
	\begin{enumerate} [label = (\roman*)]
		\item Assume that the cost satisfies
        \begin{align} \label{eqn:cost-slower-than-root-n}
            \limsup_{n \rightarrow \infty} \frac{c_n}{n^{1/2 - \varepsilon}} < \infty, \text{ for some } \varepsilon \in (0,1/2).
        \end{align}
        If we set the adjustment term $\alpha_n$ to satisfy
        \begin{align} \label{eqn:alpha-root-nlogn}
            \lim_{n \rightarrow \infty} \frac{\lvert \alpha_n \rvert}{\sqrt{n \log n}} = 0,
        \end{align}
        then the mechanism is eventually ex-ante budget balanced.
		\item Assume that the cost satisfies 
        \begin{align*}
            \lim_{n \rightarrow \infty} \frac{c_n}{n} = 0.
        \end{align*}
        If we set the adjustment term $\alpha_n$ to satisfy
        \begin{align} \label{eqn:alpha-root-n}
            \lim_{n \rightarrow \infty} \frac{\lvert \alpha_n \rvert}{\sqrt{n }} = \infty,
        \end{align}
        then the mechanism is asymptotically efficient.

		\item In particular, if the cost satisfies Condition (\ref{eqn:cost-slower-than-root-n}), and we set the adjustment term to satisfy Conditions (\ref{eqn:alpha-root-nlogn}) and (\ref{eqn:alpha-root-n}), then the AMT mechanism $(q^{n},\{t^{n}_i\})$ with transformation $h = \psi$ is both eventually ex-ante budget balanced and asymptotically efficient.
	\end{enumerate}
\end{theorem}

\begin{remark}
    The growth rate of the ex-ante budget and the convergence rate of the welfare regret ratio are provided in the proof in Appendix \ref{sec:proofs}. As an example, the adjustment term can be set as $\alpha_n = \sqrt{n \sqrt{\log n}}$, a term that simultaneously fulfills Conditions (\ref{eqn:alpha-root-nlogn}) and (\ref{eqn:alpha-root-n}).
\end{remark}

\begin{remark}
The proof of Theorem \ref{thm:hellwig} also shows that the mechanism generates an ex-ante budget surplus. This surplus grows at the rate of $\sqrt{n}\sigma_\psi\phi\left( \frac{\alpha_n}{\sqrt{n} {\sigma}_{\psi}} \right)$, where $\phi$ is the probability density function of the standard normal distribution, and $\sigma_\psi$ is the standard deviation of $\psi(V_i)$.
\end{remark}

Theorem \ref{thm:hellwig} states that if we control $\alpha_n$ to diverge at a rate between $\sqrt{n}$ and $\sqrt{n\log n}$, the mechanism can be both eventually ex-ante budget balanced and asymptotically efficient provided that the cost increases slower than $\sqrt{n}$-rate. Later, we will show that the $\sqrt{n}$-rate of cost is also necessary to balance the budget.

Before ending this section, we will briefly discuss the role of the support of $F$. The analysis in this section relies on the fact that the virtual valuation has zero mean, which is true when the lower bound of the support is exactly at zero. In general, if $F$ is supported on $[\underline{v},\bar{v}]$, then the expectation of the virtual valuation is equal to $\underline{v}$. For a small population, it is reasonable to consider a case where $\underline{v} > 0$. However, we argue that, when conducting the asymptotic analysis for $n \rightarrow \infty$ as in this paper, setting $\underline{v} = 0$ is the correct assumption. In a large economy, setting $\underline{v} = 0$ means that we do not rule out the possibility that some agents hardly benefit from the public good. 
On the other hand, assuming $\underline{v} > 0$ for a large economy means that everyone in the infinite population can receive a significant benefit from the public good. In this case, the problem becomes trivial as the ex-ante budget can be easily balanced while maintaining efficiency. 
Therefore, we assume that the lowest valuation starts from $0$, which is common in the literature \citep[e.g.,][]{hellwig2003public}.\footnote{However, we can extend our analysis to allow for a point mass at $0$ in the distribution $F$. In that case, the density $f$ is the Radon-Nykodim derivative of the probability measure $F$ with respect to the measure $L + \delta_0$, where $L$ is the Lebesgue measure on $\mathbb{R}_+$ and $\delta_0$ is the Dirac measure at $0$. This underlying measure is $\sigma$-finite, and hence our analysis applies.}

It is important to note that having $\underline{v}=0$ rather than $\underline{v}>0$ is not only reasonable but, to some extent, necessary.\footnote{We are grateful to a referee for pointing this out.} If $\underline{v}$ is strictly positive and the cost $c_n$ grows at the $\sqrt{n}$ rate, then for sufficiently large $n$, we have $n\underline{v} > c_n$. This implies that we can always provide the public good and ask each participant to share the cost. Specifically, the mechanism has a decision rule $q^n(\bm{v})=1$ and payment transfers $t_i^n(\bm{v})=c_n/n$ for all $i$. The suggested mechanism is DSIC and exactly budget balanced for any $n$, and it also satisfies ex-post efficiency and EPIR for sufficiently large $n$. This positive result in large economies offers additional justification for examining the more stringent case of $\underline{v}=0$.

It is also of interest to study the case where the lower bound $\underline{v}$ is negative. \cite{kuzmics2017public} offer real-world examples, such as a seller selling to a group or land rezoning. However, our theoretical analysis cannot be directly applied to cases involving negative valuations. The reason is that, in such instances, $\mathbb{E}[\psi(V_i)] = \underline{v}<0$. This implies that the sum $\sum \psi(V_i)$ is not centered around zero but rather at $n\underline{v} \rightarrow -\infty$. As a result, our earlier analysis of the truncated mean no longer applies. Therefore, following the approach of both \cite{mailath1990asymmetric} and \cite{hellwig2003public}, we refrain from considering cases where $\underline{v} < 0$. We defer a more comprehensive examination of this issue to future research.

\subsection{The robust mechanism} \label{ssec:robust}

There are two reasons why the mechanism examined in the previous subsection is unsatisfactory. First, the construction of such a mechanism requires the full knowledge of the virtual valuation function $\psi$. Second, when the valuation distribution is not Myerson regular, the mechanism may violate dominant-strategy incentive compatibility. We now attempt to remove these two assumptions and extend the previous results to general AMT mechanisms.

To find out when the general AMT mechanism works, we want to identify the properties of the virtual valuation function that are useful in deriving Theorem \ref{thm:hellwig} and then impose these properties on the transformation function $h$. First, we need $h$ to be increasing so that the mechanism is dominant-strategy incentive compatible. Second, we want to approximate the sum of expected payments by using a truncated normal mean as before. In this case, the total expected payment is
\begin{align} \label{eqn:total-payment-h}
	\mathbb{E}\left[ \sum t_i^{n}(\bm{V}) \right] = \mathbb{E} \left[ \sum \psi(V_i) \mathbf{1}\bigs\{ \sum (h(V_i) - \mu_{h}) \geq \alpha_n \bigs\} \right].
\end{align}
Here the total virtual valuation is truncated by the variable $\sum (h(V_i) - \mu_{h})$, which makes this total expected payment more complicated to analyze than the one in Equation (\ref{eqn:total-expected-payment-vv}). Nevertheless, most of the previous arguments can be recovered. By definition, each $h(V_i) - \mu_{h}$ is mean zero. Then by the central limit theorem, the two sums $\sum \psi(V_i)$ and $\sum (h(V_i) - \mu_{h})$ are approximately joint normal.
If $\psi(V_i)$ and $h(V_i)$ are \emph{positively correlated}, then we can infer that this truncated normal mean is positive and increases eventually towards infinity, which gives a similar result as before. However, if the correlation between $\psi(V_i)$ and $h(V_i)$ is negative, then the truncation in (\ref{eqn:total-payment-h}) becomes qualitatively different from (\ref{eqn:total-expected-payment-vv}).

Based on the above discussion, the question now becomes what additional requirements are needed for an increasing function $h$ to be positively correlated with the virtual valuation. A straightforward answer is that if $F$ is Myerson regular so that $\psi$ is increasing, then $\psi$ and $h$ are positively correlated.\footnote{For every random variable $X$ and increasing functions $g_1$ and $g_2$, the random variables $g_1(X)$ and $g_2(X)$ are positively correlated \citep[see, e.g.,][]{thorisson1995coupling}.} Perhaps more surprisingly, this positive correlation does not depend on the Myerson regularity condition; it simply results from the construction of the virtual valuation function. We summarize this result in the following lemma.

\begin{lemma} [Correlation between $\psi$ and $h$] \label{lm:corr-psi-h}

	If the function $h$ is continuous on $[0,\bar{v}]$, then the covariance between $\psi$ and $h$, denoted by $\sigma_{\psi h}$, is equal to  
	\begin{align*}
		\sigma_{\psi h} \equiv \mathbb{E}[\psi(V_i) h(V_i)] = \mathbb{E} \left[ \int_0^{V_i} v \text{ } dh(v) \right],
	\end{align*}
	where $dh(v)$ denotes integration with respect to the Lebesgue-Stieltjes measure associated with $h$. In particular, $\sigma_{\psi h} > 0$ if $h$ is increasing and continuous on $[0,\bar{v}]$.
\end{lemma}

The above lemma states that any increasing and continuous $h$ is positively correlated with the virtual valuation. As discussed above, this lemma shows that the intuition from Section \ref{ssec:hellwig} still applies when we replace $\psi$ by a general increasing function $h$. We further illustrate this result with Example \ref{eg:identity}. 

\begin{example} [continues = eg:identity]
	By using Lemma \ref{lm:corr-psi-h}, we can show that the covariance between the valuation $V_i$ and the virtual valuation $\psi(V_i)$ is in fact equal to one half of the second moment of $V_i$:
\begin{align*}
	cov(V_i,\psi(V_i)) = \mathbb{E}\left[ \int_0^{V_i} v dv \right] = \frac{1}{2}\mathbb{E}[ V_i^2 ] >0.
\end{align*} 
Therefore, the valuation is always positively correlated with the virtual valuation, even without the Myerson regularity condition. This means that if the mechanism designer observes a large $V_i$, it is likely that agent $i$'s virtual valuation is also large. Hence, it is reasonable to use the valuation in place of the virtual valuation in guiding the public-good provision decision that aims to collect payments. To the best of our knowledge, this is a new finding in regards to the literature, and we believe it is of independent research interest.
\end{example}


We now study the ex-ante budget and ex-ante welfare of the AMT mechanism $(q^{n}, \{t_i^{n}\})$ with a general transformation $h$, which are defined in Equations (\ref{eqn:AMT-budget}) and (\ref{eqn:AMT-welfare}), respectively. The following result is the main theorem of this paper, which states that the results obtained in Theorem \ref{thm:hellwig} can be generalized to any AMT mechanisms $(q^{n},\{t^{n}_i\})$ with an increasing and continuous transformation $h$. It uses the multivariate Berry-Esseen theorem \cite[e.g.,][]{bentkus2005lyapunov} to establish the convergence rate of the multivariate central limit theorem.

\begin{theorem}\label{thm:inid_irreg_welfare} 
    Assume that the transformation $h$ is increasing and continuous on $[0,\bar{v}]$, then the results in Theorem \ref{thm:hellwig} applies to the AMT mechanism with transformation $h$. In particular, if the cost satisfies Condition (\ref{eqn:cost-slower-than-root-n}), and we set the adjustment term $\alpha_n$ to satisfy Conditions (\ref{eqn:alpha-root-nlogn}) and (\ref{eqn:alpha-root-n}), then the AMT mechanism with transformation $h$ is dominant-strategy incentive compatible, ex-post individually rational, eventually ex-ante budget balanced, and asymptotically efficient.
\end{theorem}

\begin{remark}
	The growth rate of the ex-ante budget and the convergence rate of the welfare regret ratio are provided in the proof in Appendix \ref{sec:proofs}.
	An optimal convergence rate of the welfare regret ratio does not exist. 
	This is because the faster $\alpha_n$ moves toward negative infinity the faster the welfare regret ratio converges, but the asymptotic order of $\alpha_n$ needs to be strictly smaller than $\sqrt{n \log n}$ for the eventually ex-ante budget balanced condition to hold. For any mechanism $(q^{n},\{t^{n}_i\})$ that is eventually ex-ante budget balanced, we can strictly improve its welfare convergence rate by increasing the rate of $\alpha_n$ while keeping it strictly slower than $\sqrt{n \log n}$.
\end{remark}

A general AMT mechanism can be implemented with many choices of the increasing function $h$ and does not require any knowledge of the virtual valuation $\psi$ (not even the monotonicity of $\psi$ imposed by the Myerson regularity condition). In particular, as demonstrated in Example \ref{eg:identity}, the mechanism designer could simply employ the identity function. 
Furthermore, Theorem \ref{thm:inid_irreg_welfare} shows that the eventual ex-ante budget balance condition and asymptotically efficiency can still be achieved under the same assumptions on $c_n$ and $\alpha_n$ as in Theorem \ref{thm:hellwig}. 



The construction of the AMT mechanism only depends on the valuation distribution $F$ through the moment $\mu_h$. This dependence is much weaker than the case of $h = \psi$. This is because the moment $\mu_h$ is a scalar that contains much less information about $F$ than the virtual valuation function $\psi$. A moment condition can be satisfied by an infinite number of distributions which can potentially be very different, while the virtual valuation function $\psi$ uniquely determines the distribution function $F$.\footnote{To see this, note that the definition of $\psi$ implies $
d\log(1-F(v)) / dv = -1/(v-\psi(v))$.} 

We can formalize this notion of informational robustness with well-known results in statistics.
With a sample of $m$ iid observations, the moment $\mu_h$ can be estimated by the sample average with the rate of $m^{-1/2}$ under standard conditions. The estimation of the density (which is required for constructing the virtual valuation), on the other hand, is much more difficult for the following reasons. First, the density $f(v)$ can only be estimated at a slower rate of $m^{-1/3}$ for each $v \in [0,\bar{v}]$.\footnote{See, for example, \cite{stone1980optimal}. If the density $f$ is differentiable, then the optimal rate of convergence is $m^{-1/3}$. Please refer to that paper for the exact definition of optimal convergence rates.} This is due to the nonparametric nature of the density estimation problem. Second, we have to estimate the entire density function, which further slows down the convergence rate (by a $\log m$ factor when using the supremum norm).\footnote{The optimal uniform convergence rate of density estimators is derived in \cite{STONE1983optimal}.} The $m$ iid data points used for estimation are not from the current mechanism as that would distort the incentives. Instead, they can be taken from similar public projects executed in the past.

We end this subsection with a discussion of the pivot mechanism introduced in Example \ref{eg:pivot}. 
\begin{example} [continues = eg:pivot]
	The pivot mechanism implements the efficient decision rule:
\begin{align*}
	\mathbf{1} \bigs\{ \sum v_i \geq c_n \bigs\} = \mathbf{1}\bigs\{ \sum v_i \geq n \mu + c_n - n \mu \bigs\} = q^{n}(\bm{v}),
\end{align*}
with the adjustment term being $\alpha_n = c_n - n \mu $, which decreases much faster than the $\sqrt{n \log n}$ rate. Hence, the pivot mechanism runs a budget deficient, which is a well-known result in the literature. \footnote{To obtain a crude estimate, we can use Lemma \ref{lm:bound-mean-multi} in the Appendix, which indicates that the expected total payment from the pivot mechanism is at most $O\bigs(n^{1/4}\bigs)$. Therefore, the pivot mechanism runs a growing budget deficit if the cost grows faster than that rate.}
\end{example}

\subsection{Tightness of the condition on cost}\label{sec:tightness}

In both Theorem \ref{thm:hellwig} and \ref{thm:inid_irreg_welfare}, the cost $c_n$ is specified to grow slower than the $\sqrt{n}$-rate. This restriction on the growth rate of cost is not only sufficient but also, in a sense, necessary for achieving both budget balance and efficiency asymptotically. In this section we show that the optimal revenue (total expected payment) grows at the $\sqrt{n}$-rate. 
Therefore, if the cost grows faster than the $\sqrt{n}$-rate, the budget balance condition will require the provision probability to diminish as the population size grows.\footnote{Notice that, however, the cost cannot grow exactly at the $\sqrt{n}$-rate, because in that case we would need the adjustment term $\alpha_n$ to also decrease at the $\sqrt{n}$-rate to balance the budget, leading to an inefficient allocation decision.}

This result on the growth rate of the revenue is not restricted to the class of dominant-strategy incentive compatible and ex-post individually rational mechanisms. In fact, it holds for a larger class of mechanisms defined as follows. For a mechanism $(q^{n},\{t^{n}_i\})$, we say it is \emph{incentive compatible} if for each agent $i$, valuation $v_i$, and report $v_i'$,
\begin{align*} \label{eqn:incentive compatible}
	v_i \mathbb{E}[q^n(v_i,\bm{V}_{-i})] - \mathbb{E}[t^n_i(v_i,\bm{V}_{-i})] \geq v_i \mathbb{E}[q^n(v_i',\bm{V}_{-i})] - \mathbb{E}[t^n_i(v_i',\bm{V}_{-i})].
\end{align*}
This is the interim (or Bayesian) version of the dominant-strategy incentive compatibility condition specified in (\ref{eqn:dominant-strategy incentive compatible}) and is therefore weaker than the dominant-strategy incentive compatible condition. We say the mechanism is \textit{individually rational} if for each agent $i$ and valuation $v_i$,
\begin{align*} 
	v_i \mathbb{E}[q^n(v_i,\bm{V}_{-i})] - \mathbb{E}[t^n_i(v_i,\bm{V}_{-i})] \geq 0.
\end{align*}
This is the interim (or Bayesian) version of the ex-post individual rationality condition defined in (\ref{eqn:ex-post individually rational}). The class of dominant-strategy incentive compatible and ex-post individually rational mechanisms is contained in the class of incentive compatible and individually rational mechanisms.

The following theorem shows that the optimal revenue for the class of incentive compatible and individually rational mechanisms grows at the $\sqrt{n}$-rate.

\begin{theorem} [Tightness of the Condition on Cost] \label{thm:tightness}
	The maximum total expected payment for any sequence of incentive compatible and individually rational mechanisms grows at the $\sqrt{n}$-rate. That is,
	\begin{align*}
		\limsup_{n \rightarrow \infty} \sup_{(q^n,\{t_i^n\})} \mathbb{E} \left[ \sum t^n_i(\bm{V}) \right] \big/ \sqrt{n} < \infty,
	\end{align*}
    where the supremum $\sup_{(q^n,\{t_i^n\})}$ is taken over all incentive compatible and individually rational mechanisms.
	Consequently, if the cost satisfies
    \begin{align*}
        \lim_{n \rightarrow \infty} \frac{c_n}{\sqrt{n}} = \infty,
    \end{align*}
    then for any sequence of mechanisms $(q^n,\{t_i^n\})$ that is incentive compatible, individually rational, and eventually ex-ante budget balanced, the provision probability of the public good converges to zero under the following rate
	\begin{align*}
		\limsup_{n \rightarrow \infty} \frac{\mathbb{P}(q^n(\bm{V}) = 1)}{\sqrt{n}/c_n} < \infty .
	\end{align*}
\end{theorem}

\begin{remark}
	As pointed out by \cite{mailath1990asymmetric} in their Theorem 2, if $c_n \leq n(\mu - \delta)$ for some $\delta \in (0,\mu)$, then the efficient (first-best) provision probability should converge to one. In such circumstances, our Theorem \ref{thm:tightness} implies that no incentive compatible, individually rational and eventually ex-ante budget balanced mechanism can be asymptotically efficient when $c_n$ grows faster than the $\sqrt{n}$-rate.
\end{remark}

Theorem \ref{thm:tightness} is an enhancement of the results in \cite{mailath1990asymmetric} in two ways. First, in that paper, the authors show that the provision probability converges to zero when the cost grows proportionally with the number of agents. Here, we show that any rate of $c_n$ faster than $\sqrt{n}$ gives the same negative result. Second, \cite{mailath1990asymmetric} illustrate that the provision probability converges at the $n^{-1/4}$-rate when the cost grows proportionally with the number of agents. This convergence rate can be refined to the faster $n^{-1/2}$-rate using our result in Theorem \ref{thm:tightness}. 


The $\sqrt{n}$-rate of the revenue given by Theorem \ref{thm:tightness} is known in the literature \citep{hellwig2003public,kleinberg2013ratio}.\footnote{As we clarify in Appendix \ref{sec:comparison-hellwig}, while the result in \cite{hellwig2003public} are correct, the proof is slightly flawed. On the other hand, \cite{kleinberg2013ratio} only provide a heuristic (but informal) argument with the uniform distribution to explain why the optimal revenue grows at the $\sqrt{n}$-rate.}
This tightness result, together with the positive result in Theorem \ref{thm:inid_irreg_welfare},  characterizes the role of the provision cost in the public goods problem. In particular, we identify $\sqrt{n}$ as the ``boundary case'' of the growth rate of cost. When the cost grows faster than the $\sqrt{n}$-rate, we can recover the negative result by \cite{mailath1990asymmetric} that no incentive compatible, individually rational, and eventually ex-ante budget balanced mechanism can be asymptotically efficient. When the cost grows slower than the $\sqrt{n}$-rate, we can recover the positive result by \cite{hellwig2003public} that budget balance and efficiency can be achieved asymptotically (and in an informationally robust way).\footnote{The nature of the assumption on the growth rate of $c_n$ is essentially about the production technology and how the marginal cost decreases with the number of agents. See \citet{roberts1976incentives} for a discussion on the relationship between the cost of a public good and the number of consumers.}

\section{Concluding Remarks} \label{sec:conlusion}

The main objective of this paper is to address the problem of mechanism design for public goods provision when the number of agents grows to infinity. The public goods problem is a classic textbook problem, yet the asymptotic and statistical approach we take is novel. The technical contribution of the paper is in introducing the central limit estimate and the Berry-Esseen theorem as helpful tools in calculating quantities such as ex-ante budget and revenue, which was previously a formidable task. We show that what the mechanism designer can achieve depends crucially on how the cost of providing the public good grows with the number of agents. When the cost increases faster than the $\sqrt{n}$-rate, the mechanism designer cannot implement the efficient decision, even asymptotically, with a balanced budget. When the cost grows slower than the $\sqrt{n}$-rate, we advocate the mechanism designer to use AMT mechanisms. The AMT mechanisms can achieve budget balance and efficiency asymptotically and have the advantage of being simple and informationally robust.

\begin{appendix}
    \section{Technical Proofs}\label{sec:proofs}
    Appendix \ref{sec:lemmas} presents general convergence results dervied based on the Berry-Esseen theorem. Appendix \ref{sec:proof-thm} presents the proofs for theorems in the main text.

    \subsection{Preliminary Results}\label{sec:lemmas}
    This section provides some general results as useful lemmas. We denote $\phi$ and $\Phi$ respectively as the pdf and cdf of the standard normal distribution. We use $C$ to denote a generic constant that does not depend on $n$, which may have different values at each appearance.
    
    \begin{lemma} \label{lm:truncated-mean-as-tail-prob}
        Let $X$ and $Y$ be two random variables with $\mathbb{E}\abs{X} < \infty$, then for any $\alpha \in \mathbb{R}$,
        \begin{align*}
            \mathbb{E} \left[ X \mathbf{1}_{[\alpha,\infty)}(Y) \right] = \int_{0}^\infty \mathbb{P}\left( X > x,Y \geq \alpha \right) dx - \int_{-\infty}^0 \mathbb{P}\left( X < x,Y \geq \alpha \right) dx.
        \end{align*}
        In particular, if $X=Y$, then
        \begin{align*}
            \mathbb{E} \left[ X \mathbf{1}_{[\alpha,\infty)}(X) \right] = \int_{0}^\infty \mathbb{P}\left( X >x \right) dx - \int_{\alpha}^0 \mathbb{P}\left( X < x\right) dx, \text{ for } \alpha \leq 0,
        \end{align*}
        and
        \begin{align*}
            \mathbb{E} \left[ X \mathbf{1}_{[\alpha,\infty)}(X) \right] = \int_{\alpha}^\infty \mathbb{P}\left( X >x \right) dx, \text{ for } \alpha > 0.
        \end{align*}
    \end{lemma}
    
    \begin{proof} [Proof of Lemma \ref{lm:truncated-mean-as-tail-prob}]
        Define $X^+ = \max\{X,0\}$ and $X^- = \max\{-X,0\}$. For $X^+$, notice that
        \begin{align*}
            X^+ = \int_0^\infty \mathbf{1}_{[0,X^+]}(x) dx.
        \end{align*}
        Using Fubini theorem, we have
        \begin{align*}
            \mathbb{E} \left[ X^+ \mathbf{1}_{[\alpha,\infty)}(Y) \right] & = \mathbb{E} \left[ \int_0^\infty \mathbf{1}_{[0,X^+]}(x) \mathbf{1}_{[\alpha,\infty)}(Y) dx  \right] \\
            & = \int_0^\infty \mathbb{E} \left[ \mathbf{1}_{[0,X^+]}(x) \mathbf{1}_{[\alpha,\infty)}(Y)  \right]dx \\
            & = \int_0^\infty \mathbb{P} \left( X^+ > x, Y \geq \alpha  \right)dx \\
            & = \int_0^\infty \mathbb{P} \left( X > x, Y \geq \alpha  \right)dx.
        \end{align*}
        Similarly, for $X^-$, we have
        \begin{align*}
            \mathbb{E} \left[ X^- \mathbf{1}_{[\alpha,\infty)}(Y) \right] 
            & = \int_0^\infty \mathbb{P} \left( X^- > x, Y \geq \alpha  \right)dx \\
            & = \int_{-\infty}^0 \mathbb{P} \left( X < x, Y \geq \alpha  \right)dx.
        \end{align*}
        The result then follows from the fact that $\mathbb{E}[X] = \mathbb{E}[X^+] - \mathbb{E}[X^-]$.
    \end{proof}
    
    Let $(X_i,Y_i), 1 \leq i \leq n$ be an independent sequence of random vectors in $\mathbb{R}^2$ with zero mean and $\mathbb{E}\abs{X_i}^3,\mathbb{E}\abs{Y_i}^3 < \infty$. We introduce the following set of notations for their marginal moments: $\sigma^2_{X} = \mathbb{E}X_i^2,  \sigma_{Y} = \mathbb{E}Y_i^2,\rho_{X} = \mathbb{E}\abs{X_i}^3,  \rho_{Y} = \mathbb{E}\abs{Y_i}^3,$
    and correlation: $\sigma_{XY} = \mathbb{E}[X_iY_i].$
    All these moments are finite. Define $S^X_n = \sum X_i$ and $S^Y_n = \sum Y_i$. Let $(Z^X,Z^Y)$ be a joint normal random vector with zero mean and the same covariance structure as $(S^X_n,S^Y_n)$. We use $\norm{\cdot}$ to denote both the induced 2-norm of matrices and the Euclidean norm of vectors.
    
    We want to bound the difference between the distributions of $(S^X_n,S^Y_n)$ and $(Z^X,Z^Y)$ using the Berry-Esseen theorem. The following lemma is the univariate Berry-Esseen theorem.
    
    \begin{lemma}\label{lm:berry-esseen-uni}
        There is a constant $C > 0$ such that
        \begin{align*}
            \sup_{x \in \mathbb{R}} \abs{ \mathbb{P}\left( S_n^X \leq x \right) - \Phi\left( \frac{x}{\sqrt{n}
            {\sigma}_{X}} \right)  } \leq \frac{C}{\sqrt{n}} .
        \end{align*}
    \end{lemma}
    
    Then we prove the multivariate case.
    
    \begin{lemma} \label{lm:berry-esseen-multi}
        Assume $X$ and $Y$ are not perfectly correlated. The following bound holds between the joint distributions of $(S^X_n,S^Y_n)$ and $(Z^X,Z^Y)$: let $\mathcal{B}$ be the set of all measurable convex sets in $\mathbb{R}^2$, then there is a constant $C>0$ such that
        \begin{align*}
            & \sup_{B \in \mathcal{B}} \abs{\mathbb{P}\left( (S^X_n,S^Y_n) \in B \right)  - \mathbb{P}\left((Z^X,Z^Y) \in B\right)} \leq \frac{C}{\sqrt{n}}.
        \end{align*}
    \end{lemma}
    
    \begin{proof} [Proof of Lemma \ref{lm:berry-esseen-multi}]
        We use $\varSigma_{XY}$ to denote the normalized (by $n$) covariance matrix of $(S^X_n,S^Y_n)$, 
        \[
        \varSigma_{XY}	=
        \begin{pmatrix*}
            {\sigma}^2_{X} & {\sigma}_{XY} \\
            {\sigma}_{XY}  & {\sigma}^2_{Y}
        \end{pmatrix*}.
        \]
        The multivariate Berry-Esseen theorem \citep{bentkus2005lyapunov} says there exists a universal constant $C$, such that
        \begin{align*}
            & \sup_{B \in \mathcal{B}} \abs{\mathbb{P}\left( (S^X_n,S^Y_n) \in B \right)  - \mathbb{P}\left((Z^X,Z^Y) \in B\right)} \\
             & \leq  C \sum \mathbb{E} \left[ \norm{(n\varSigma_{XY})^{-1/2} (X_i,Y_i)'}^3 \right] \\
            & \leq  \frac{C}{\sqrt{n}} \left( \frac{1}{n} \sum \mathbb{E} \left[ \norm{\varSigma_{XY}^{-1/2}}^3 \norm{(X_i,Y_i)'}^3 \right] \right) \\
            & = \frac{C}{\sqrt{n}} \norm{\varSigma_{XY}^{-1/2}}^3 \mathbb{E} \left[  \norm{(X_i,Y_i)'}^3 \right].
        \end{align*}
        The induced 2-norm of a positive semi-definite matrix equals to its largest eigenvalue. So we have
        $$ \norm{\varSigma_{XY}^{-1/2}}^3 = \lambda_{\text{min}}^{-3/2}, $$
        where $\lambda_{\text{min}}$ is the smallest eigenvalue of $\varSigma_{XY}$. We next compute $\lambda_{\text{min}}$. The characteristic function of $\varSigma_{XY}$ is
        \begin{align*}
            \text{det} 
            \begin{pmatrix*}
                \bar{\sigma}^2_{X} - \lambda & {\sigma}_{XY} \\
                {\sigma}_{XY}  &  {\sigma}^2_{Y} - \lambda
            \end{pmatrix*}
            = \lambda^2 - ({\sigma}^2_{X} + {\sigma}^2_{Y}) \lambda + {\sigma}^2_{X}{\sigma}^2_{Y} - {\sigma}_{XY}^2.
        \end{align*}
        The smaller eigenvalue is
        \begin{align*}
            \lambda_{\text{min}} = & \frac{1}{2} \left( {\sigma}^2_{X} + {\sigma}^2_{Y} - \sqrt{({\sigma}^2_{X} + {\sigma}^2_{Y})^2 - 4({\sigma}^2_{X}{\sigma}^2_{Y} - {\sigma}_{XY}^2)} \right) \\
            = & \frac{1}{2} \left( {\sigma}^2_{X} + {\sigma}^2_{Y} - \sqrt{({\sigma}^2_{X} - {\sigma}^2_{Y})^2 - 4{\sigma}_{XY}^2} \right),
        \end{align*}
        which is a positive constant when $\sigma^2_X\sigma^2_Y - \sigma_{XY}^2 \neq 0$.

        To deal with the remaining part of the bound, we employ the Minkowski inequality.
        \begin{align*}
            \mathbb{E} \left[  \norm{(X_i,Y_i)'}^3 \right] & = \mathbb{E} \left[  (X_i^2 + Y_i^2)^{3/2} \right] \\
            & \leq \left( \left( \mathbb{E}\left[ (X_i^2)^{3/2} \right] \right)^{2/3} + \left( \mathbb{E}\left[ (Y_i^2)^{3/2} \right] \right)^{2/3} \right)^{3/2} \\
            & = \left( \left( \mathbb{E} \abs{X_i}^{3}\right)^{2/3} + \left( \mathbb{E}\abs{Y_i}^{3}  \right)^{2/3} \right)^{3/2} \\
            & = \left( \rho_{X}^{2/3} + \rho_{Y}^{2/3} \right)^{3/2} \\
            & \leq \sqrt{2}(\rho_{X} + \rho_{Y}).
        \end{align*}
        The last inequality follows from the fact that the function $x \mapsto x^{2/3}$ is concave so that for any two positive numbers $a$ and $b$, it holds that
        \begin{align*}
            \frac{a^{2/3}+b^{2/3}}{2} \leq \left( \frac{a+b}{2} \right)^{2/3}  \implies \left( a^{2/3}+b^{2/3} \right)^{3/2} \leq \sqrt{2}(a+b). 
        \end{align*}
    \end{proof}
    
    \begin{lemma} \label{lm:multinormal-truncated-mean}
        The following expression of the normal truncated mean holds true:
        \begin{align*}
            \mathbb{E}\left[ Z^X \mathbf{1}\{ Z^Y \geq \alpha_n \} \right] = \sqrt{n} \frac{{\sigma}_{XY}}{{\sigma}_{Y}} \phi\left( \frac{\alpha_n}{\sqrt{n} {\sigma}_{Y}} \right).
        \end{align*}
        In particular, for $X = Y$, we have
        \begin{align*}
            \mathbb{E}\left[ Z^X \mathbf{1}\{ Z^X \geq \alpha_n \} \right] = \sqrt{n} {\sigma}_{X} \phi\left( \frac{\alpha_n}{\sqrt{n} {\sigma}_{X}} \right).
        \end{align*}
    \end{lemma}
    
    \begin{proof} [Proof of Lemma \ref{lm:multinormal-truncated-mean}]
        Define $e = Z^X - \frac{{\sigma}_{XY}}{{\sigma}^2_{Y}} Z^Y$. It is straightforward to compute that $\mathbb{E}e=0$. Since $e$ and $Z^Y$ are jointly normal and uncorrelated, $e \perp Z^Y$. Therefore,
        \begin{align*}
            \mathbb{E}\left[ Z^X \mathbf{1}{\{ Z^Y \geq \alpha_n \}} \right] & = \mathbb{E}\left[ \left( \frac{{\sigma}_{XY}}{{\sigma}^2_{Y}} Z^Y + e \right)  \mathbf{1}\{ Z^Y \geq \alpha_n \} \right] \\
            & = \frac{{\sigma}_{XY}}{{\sigma}^2_{Y}} \mathbb{E}\left[ Z^Y \mathbf{1}\{ Z^Y \geq \alpha_n \} \right] \\
            & = \sqrt{n} \frac{{\sigma}_{XY}}{{\sigma}_{Y}} \phi\left( \frac{\alpha_n}{\sqrt{n} {\sigma}_{Y}} \right).
        \end{align*}
    \end{proof}
    
    \begin{lemma} \label{lm:bound-mean-uni}
        For any sequence of numbers $\alpha_n$, there is a constant $C>0$ such that
        \begin{align*}
            \abs{\mathbb{E}\left[ S_n^X \mathbf{1}\{ S^X_n \geq \alpha_n \} \right] - \sqrt{n} {\sigma}_{X}\phi\left( \frac{\alpha_n}{\sqrt{n} {\sigma}_{X}} \right)} \leq C n^{1/4} .
        \end{align*}
    \end{lemma}
    
    \begin{proof} [Proof of Lemma \ref{lm:bound-mean-uni}]
        By Lemma \ref{lm:truncated-mean-as-tail-prob} and \ref{lm:multinormal-truncated-mean}, we have
        \begin{align*}
            &\abs{\mathbb{E}\left[ S_n^X \mathbf{1}\{ S^X_n \geq \alpha_n \} \right] - \sqrt{n} {\sigma}_{X}\phi\left( \frac{\alpha_n}{\sqrt{n} {\sigma}_{X}} \right)} \\
            \leq & \int_{0}^\infty \abs{\mathbb{P}\left( S^X_n > x \right) - \mathbb{P}\left( Z^X > x \right)} dx + \int_{-\infty}^0 \abs{\mathbb{P}\left( S^X_n < x\right) - \mathbb{P}\left( Z^X < x\right)} dx.
        \end{align*}
        Using Chebyshev's inequality, we have
        \begin{align*}
            \abs{\mathbb{P}\left( S^X_n > x \right) - \mathbb{P}\left( Z^X > x \right)} \leq \frac{n {\sigma}^2_{X}}{x^2}.
        \end{align*}
        Together with Lemma \ref{lm:berry-esseen-uni}, we have
        \begin{align*}
            \int_{0}^\infty \abs{\mathbb{P}\left( S^X_n > x \right) - \mathbb{P}\left( Z^X > x \right)} dx & \leq \int_{0}^{n^{3/4}} \frac{C}{\sqrt{n}} \frac{{\rho}_{X}}{{\sigma}_{X}^{3}} dx + \int_{n^{3/4}}^\infty \frac{n {\sigma}^2_{X}}{x^2} dx \\
            & = n^{1/4} \left( C \frac{{\rho}_{X}}{{\sigma}_{X}^{3}} + {\sigma}_{X}^2 \right) 
        \end{align*}
    \end{proof}
    
    \begin{lemma} \label{lm:bound-mean-multi}
        Suppose $X$ and $Y$ are not perfectly correlated. For any sequence of numbers $\alpha_n$, there is a constant $C>0$ such that
        \begin{align*}
            \abs{\mathbb{E}\left[ S_n^X \mathbf{1}\{ S^Y_n \geq \alpha_n \} \right] - \sqrt{n} \frac{{\sigma}_{XY}}{{\sigma}_{Y}} \phi\left( \frac{\alpha_n}{\sqrt{n} {\sigma}_{Y}} \right)}\leq C n^{1/4}.
        \end{align*}
    \end{lemma}
    
    \begin{proof} [Proof of Lemma \ref{lm:bound-mean-multi}]
        By Lemma \ref{lm:truncated-mean-as-tail-prob} and \ref{lm:multinormal-truncated-mean}, we know that the absolute difference we want to bound is
        \begin{align*}
            & \abs{\mathbb{E}\left[ S_n^X \mathbf{1}\{ S^Y_n \geq \alpha_n \} \right] - \mathbb{E}\left[ Z^X \mathbf{1}\{ Z^Y \geq \alpha_n \} \right]} \\
            \leq & \underbrace{\abs{ \int_{0}^\infty \mathbb{P}\left( S^X_n > x,S_n^Y \geq \alpha_n \right) - \mathbb{P}\left( Z^X > x,Z^Y \geq \alpha_n \right) dx}}_{I_1} \\
            + & \underbrace{\abs{\int_{-\infty}^0 \mathbb{P}\left( S^X_n < x,S^Y_n \geq \alpha_n \right) - \mathbb{P}\left( Z^X < x,Z^Y \geq \alpha_n \right) dx }}_{I_2}.
        \end{align*}
        We first deal with $I_1$. Notice that 
        \begin{align*}
            \mathbb{P}\left( S^X_n > x,S_n^Y \geq \alpha_n \right) \leq \mathbb{P}\left( S^X_n > x\right) \leq \frac{n{\sigma}^2_{X}}{x^2},
        \end{align*}
        where the second inequality follows from the Chebyshev inequality. Similarly, we have
        \begin{align*}
            \mathbb{P}\left( Z^X > x,Z^Y \geq \alpha_n \right) \leq \mathbb{P}\left( Z^X > x\right) \leq \frac{n{\sigma}^2_{X}}{x^2}.
        \end{align*}
        So putting these two tail bounds together, we have
        \begin{align*}
            \abs{\mathbb{P}\left( S^X_n > x,S_n^Y \geq \alpha_n \right) - \mathbb{P}\left( Z^X > x,Z^Y \geq \alpha_n \right)} \leq \frac{n{\sigma}^2_{X}}{x^2}, \text{ for all } x>0.
        \end{align*}
        Then, based on Lemma \ref{lm:berry-esseen-multi}, we have
        \begin{align*}
            I_1 &\leq \int_{0}^{n^{3/4}} \abs{ \mathbb{P}\left( S^X_n > x,S_n^Y \geq \alpha_n \right) - \mathbb{P}\left( Z^X > x,Z^Y \geq \alpha_n \right)} dx \\
            & + \int_{n^{3/4}}^\infty \abs{ \mathbb{P}\left( S^X_n > x,S_n^Y \geq \alpha_n \right) - \mathbb{P}\left( Z^X > x,Z^Y \geq \alpha_n \right)} dx \\
            & \leq n^{3/4} \frac{C}{\sqrt{n}} + \int_{n^{3/4}}^\infty \frac{n{\sigma}^2_{X}}{x^2}  dx \\
            & = n^{1/4} \left( C + {\sigma}^2_{X} \right).
        \end{align*}
        Following the same steps we can derive a same bound for $I_2$. Then the result follows.
    \end{proof}

    \subsection{Proofs of Results in the Main Text}\label{sec:proof-thm}

    We introduce some notations that are useful in the proofs. The moments of the valuation distribution are denoted by 
    \begin{align*}
        \mu \equiv \mathbb{E}[V_i], \sigma^2 \equiv \mathbb{E}\vert V_i - \mu \vert^2, \rho \equiv \mathbb{E}\vert V_i - \mu \vert^3, \sigma_{\psi}^2 \equiv \mathbb{E}\vert \psi(V_i) \vert^2, \rho_{\psi} \equiv \mathbb{E}\vert \psi(V_i) \vert^3.
    \end{align*}
    They are all finite since we assume that $f$ has a bounded support and is bounded away from zero on the support. 
To keep the asymptotic analysis concise, we use the notations in the following table.\footnote{The first five notations are standard. The expression $d_n = O_\varepsilon(e_n)$ means that $d_n$ is asymptotically bounded above by $e_n$ divided by some (sufficiently small) power of $n$. The expression $d_n = \omega_{\varepsilon}(e_n)$ means that $d_n$ asymptotically dominates $e_n$ divided by any power of $n$.}
\begin{table}[!hbtp]
	\centering
	
	\begin{tabular}{lll}
	\toprule
	 Notation & Definition & Short Explanation \\ \midrule
		$d_n = o(e_n)$ & $d_n/e_n \rightarrow 0 $ & $\lvert d_n \rvert$ dominated by $e_n$ \\
		$d_n = O(e_n)$ & $\limsup \abs{d_n}/e_n < \infty $ & $\lvert d_n \rvert$ bounded above by $e_n$  \\
		$d_n = \omega(e_n)$ & $\lvert d_n \rvert/e_n \rightarrow \infty$ & $\lvert d_n \rvert$ dominates $e_n$ \\
		$d_n = \Omega(e_n)$ & $ \liminf \lvert d_n \rvert/e_n > 0$ & $\lvert d_n \rvert$ bounded below by $e_n$ \\
		$d_n = \Theta(e_n)$ & $ \lvert d_n \rvert = O(e_n),\lvert d_n \rvert = \Omega(e_n)$ & $\lvert d_n \rvert$ bounded below and above by $e_n$ \\
		$d_n = O_\varepsilon(e_n)$ & $ \exists \varepsilon>0, d_n = O(n^{-\varepsilon}e_n)$ & $\lvert d_n \rvert$ nearly bounded above by $e_n$ \\
		$d_n = \omega_{\varepsilon}(e_n)$ & $ \forall \varepsilon>0, d_n = \omega(n^{-\varepsilon}e_n)$ & $\lvert d_n \rvert$ nearly dominates $e_n$ \\
	\bottomrule
	\end{tabular}

\end{table}

    \begin{proof} [Proof of Theorem \ref{thm:hellwig}]
        \begin{enumerate}[label = (\roman*)]
            \item We first obtain an asymptotic lower bound for the ex-ante budget:
            \begin{align*}
                \mathbb{E}[b^{n}(\bm{V})] & = \mathbb{E} \left[ \sum \psi(V_i) \mathbf{1}\big\{ \sum \psi(V_i) \geq \alpha_n \big\} \right] - c_n \mathbb{P}\left( \sum \psi(V_i) \geq \alpha_n \right) \\
                & \geq \mathbb{E} \left[ \sum \psi(V_i) \mathbf{1}\big\{ \sum \psi(V_i) \geq \alpha_n \big\} \right] - c_n \\
                & \geq \sqrt{n} {\sigma}_{\psi} \phi\left( \frac{\alpha_n}{\sqrt{n} {\sigma}_{\psi}} \right) - Cn^{1/4}  - c_n,
            \end{align*}
            where the last inequality follows from Lemma \ref{lm:bound-mean-uni}. Then we show that the first term on the RHS is the leading term. First, by the assumption $\alpha_n = o\bigs( \sqrt{ n \log n }  \bigs)$, we have
            \begin{align*}
                \log\left( \sqrt{n} \phi\left( \frac{\alpha_n}{\sqrt{n} {\sigma}_{\psi}} \right) \Big/ n^{1/4}  \right) & = \frac{1}{4}\log n  -\left( \frac{\alpha_n}{{\sigma}_{\psi}\sqrt{n}} \right)^2 + C  \\
                & = \log n \left( \frac{1}{4} - \left( \frac{\alpha_n}{{\sigma}_{\psi}\sqrt{n \log n}} \right)^2 + o(1)\right) \\
                & = \log n (1/4 + o(1))\rightarrow \infty. 
            \end{align*}
            
            Next, by the assumption that $c_n = O_{\varepsilon}\bigs(\sqrt{n}\bigs)$, there exist $C,\varepsilon>0$ such that $c_n \leq C n^{1/2 - \varepsilon}$ for large $n$. Then $\log c_n / \log n \leq 1/2 - \varepsilon /2$ for large $n$. Therefore,
            \begin{align*}
                & \log\left( \sqrt{n} \phi\left( \frac{\alpha_n}{\sqrt{n} {\sigma}_{\psi}} \right) \Big/ c_n \right) \\
                = & \frac{1}{2} \log n  - \left( \frac{\alpha_n}{{\sigma}_{\psi}\sqrt{n}} \right)^2 - \log c_n  + C \\
                = & \log n \left( \frac{1}{2} - \left( \frac{\alpha_n}{{\sigma}_{\psi}\sqrt{n \log n}} \right)^2 - \frac{\log c_n}{\log n} + o(1) \right) \\
                \geq & \log n (\varepsilon/2 + o(1)) \rightarrow \infty.
            \end{align*}
            The above asymptotic results show that the leading term in the lower bound is $\sqrt{n} {\sigma}_{\psi} \phi\left( \frac{\alpha_n}{\sqrt{n} {\sigma}_{\psi}} \right)$. We then derive an upper bound for the ex-ante budget:
            \begin{align*}
                \mathbb{E}[b^{n}(\bm{V})] & \leq \mathbb{E} \left[ \sum \psi(V_i) \mathbf{1}\big\{ \sum \psi(V_i) \geq \alpha_n \big\} \right] \\
                & \leq \sqrt{n} {\sigma}_{\psi} \phi\left( \frac{\alpha_n}{\sqrt{n} {\sigma}_{\psi}} \right) + Cn^{1/4} ,
            \end{align*}
            where the last inequality again follows from Lemma \ref{lm:bound-mean-uni}. By the previous analysis, the leading term in the upper bound is also $\sqrt{n}{\sigma}_{\psi} \phi\left( \frac{\alpha_n}{\sqrt{n} {\sigma}_{\psi}} \right)$. Therefore, we only need to derive the growth rate of that term. Take any $\varepsilon>0$, we have
            \begin{align*}
                \log \left( \sqrt{n} \phi\left( \frac{\alpha_n}{\sqrt{n} {\sigma}_{\psi}} \right) \Big/ n^{1/2-\varepsilon} \right) & = \varepsilon \log n - \frac{1}{2} \left( \frac{\alpha_n}{{\sigma}_{\psi}\sqrt{n}} \right)^2 +C \\
                & = \log n \left(\varepsilon - \frac{1}{2} \left( \frac{\alpha_n}{{\sigma}_{\psi}\sqrt{n \log n}} \right)^2 + o(1)  \right) \\
                & = \log n \left( \varepsilon + o(1) \right) \rightarrow \infty.  
            \end{align*}

            \item 
        First notice that $\mathbb{E}[w^{n}(\bm{V})] \leq \mathbb{E}[w^*(\bm{V})]$. We break the ex-ante welfare into two parts
        \begin{align*}
            \mathbb{E}[w^{n}(\bm{V})] & = \mathbb{E} \left[ \sum V_i q^{n}(\bm{V}) - \sum t_i^{n}(\bm{V}) \right] \\
            & = \mathbb{E} \left[ \left( \sum V_i  -  c_n \right) q^{n}(\bm{V})\right] - \mathbb{E}[b^{n}(\bm{V})].
        \end{align*}
        We obtain a lower bound for the first term on the RHS:
        \begin{align*}
            & \mathbb{E} \left[ \left( \sum V_i  -  c_n \right) q^{n}(\bm{V})\right] \\
            = & \mathbb{E} \left[ \mathbf{1}\big\{ \sum \psi(V_i) \geq \alpha_n \big\} \sum V_i  \right] - c_n \mathbb{P}\left( \sum \psi(V_i) \geq c_n \right) \\
            \geq & \mathbb{E} \left[ \mathbf{1}\big\{ \sum \psi(V_i) \geq \alpha_n \big\} \right] \mathbb{E}\left[ \sum V_i \right] - c_n \mathbb{P}\left( \sum \psi(V_i) \geq c_n \right)  \\
            = & (n {\mu}- c_n) \mathbb{P}\left( \sum \psi(V_i) \geq \alpha_n \right) \\
            \geq & (n {\mu}- c_n) \left( 1-\Phi\left( \frac{\alpha_n}{\sqrt{n} {\sigma}_{\psi}} \right) - \frac{C}{\sqrt{n}} \right)\\
            \geq & n {\mu}- c_n - n {\mu}\left( \Phi\left( \frac{\alpha_n}{\sqrt{n}{\sigma}_{\psi}} \right) + \frac{C}{\sqrt{n}}\right) ,
        \end{align*}
        where the second inequality follows from the fact that $\psi$ is non-decreasing and the second to last inequality follows from Lemma \ref{lm:berry-esseen-uni}. By the result in part (i), we can bound the ex-ante budget from above by
        \begin{align*}
            \mathbb{E}[b^{n}(\bm{V})] & \leq \sqrt{n}{\sigma}_{\psi} \phi\left( \frac{\alpha_n}{\sqrt{n} {\sigma}_{\psi}} \right) + Cn^{1/4}\\
            & \leq \sqrt{n} {\sigma}_{\psi} / \sqrt{2 \pi} + Cn^{1/4} ,
        \end{align*}
        where the last inequality follows from the fact that $\phi(\cdot) \leq 1/\sqrt{2 \pi}$. The ex-ante efficient welfare is bounded above by $\mathbb{E}[w^*(\bm{V})] \leq \mathbb{E}\left[ \sum V_i \right] = n {\mu}$.
        Combining these results, we get
        \begin{align*}
            \frac{\mathbb{E}[w^*(\bm{V})] - \mathbb{E}[w^{n}(\bm{V})]}{\mathbb{E}[w^*(\bm{V})]} & \leq \frac{c_n}{n\mu} + \Phi\left( \frac{\alpha_n}{\sqrt{n} {\sigma}_{\psi}} \right) + \frac{C}{\sqrt{n}} \rightarrow 0. \\
        \end{align*}
        
        The remaining task is to show that the leading term in the above expression is $\Phi\bigs( \alpha_n / \sqrt{n} {\sigma}_{\psi} \bigs)$. 
        We use the well-known lower bound on the tail of $\Phi$:
        \begin{align*}
            \Phi(v) = 1-\Phi(|v|) \geq  \abs{v}\exp(-v^2/2)/\big(\sqrt{2\pi}(1+v^2)\big), v \leq 0.
        \end{align*} 
        Then
        \begin{align*}
            &\log\left( \sqrt{n} \Phi\left( \frac{\alpha_n}{{\sigma}_{\psi}\sqrt{n}} \right)  \right)\\
            \geq & \frac{1}{2} \log n + \log \left( \frac{\abs{\alpha_n}}{\sqrt{n}}  \right) - \frac{1}{2} \left( \frac{\alpha_n}{{\sigma}_{\psi}\sqrt{n}} \right)^2 - \log\left( 1 + (\alpha_n/{\sigma}_{\psi}\sqrt{n})^2 \right) + C \\
            = & \frac{1}{4} \log n + \log \left( \frac{\abs{\alpha_n}}{\sqrt{n}}  \right)- \frac{1}{2} \left( \frac{\alpha_n}{{\sigma}_{\psi}\sqrt{n}} \right)^2 + \log\left( \frac{n^{1/4}}{1 + (\alpha_n/{\sigma}_{\psi}\sqrt{n})^2} \right)  + C \\
            = & \log n \left( \frac{1}{4} - \frac{1}{2} \left( \frac{\alpha_n}{{\sigma}_{\psi}\sqrt{n \log n}} \right)^2 \right) + \log \left( \frac{\abs{\alpha_n}}{\sqrt{n}}  \right)\\
            + & \log\left( \frac{n^{1/4}/\log n}{1/ \log n + (\alpha_n/{\sigma}_{\psi}\sqrt{n \log n})^2} \right)  + C,
        \end{align*}
        where the last line goes to $\infty$. Therefore, $\Phi\bigs( \alpha_n / \sqrt{n} {\sigma}_{\psi} \bigs)$ is $\omega(1/\sqrt{n})$ and hence the leading term of the welfare ratio.

        \item The result comes directly from the previous two parts.

    \end{enumerate}
    \end{proof}

    \begin{proof} [Proof of Lemma \ref{lm:corr-psi-h}]
        By the definition of $\psi$, we have
        \begin{align*}
            \sigma_{\psi h} & = \mathbb{E} \left[ \psi(V_i) h(V_i) \right] \\
            & = \mathbb{E} \left[ V_i h(V_i) \right] - \int_0^\infty h(v) (1-F(v)) dv.
        \end{align*}
        By the Fubini theorem, the second term on the RHS is
        \begin{align*}
            \int_0^\infty h(v) (1-F(v)) dv & = \int_0^\infty h(v) \mathbb{E}[\mathbf{1}_{[0,V_i)}(v)] dv  \\
            & = \mathbb{E}\left[ \int_0^\infty h(v) \mathbf{1}_{[0,V_i)}(v) dv \right]  \\
            &= \mathbb{E} \left[ \int_0^{V_i} h(v) dv \right] < \infty.
        \end{align*}
        The above quantity is finite since $V$ has a bounded support and $h$ is continuous.
        Then we perform integration by parts to the above Lebesgue-Stieltjes integral:
        \begin{align*}
            \int_0^{V_i} h(v) dv = vh(v)\Big\lvert_{0}^{V_i} - \int_0^{V_i} v dh(v) = V_ih(V_i) - \int_0^{V_i} v  dh(v).
        \end{align*}
        Therefore, 
        \begin{align*}
            \sigma_{\psi h} & = \mathbb{E} \left[ V_ih(V_i) - \int_0^{V_i} h(v) dv \right] = \mathbb{E} \left[\int_0^{V_i} v \text{ } dh(v)\right],
        \end{align*}
        which is positive if $h$ is increasing.
    \end{proof}
    
    For Theorem \ref{thm:inid_irreg_welfare}, we introduce the following notations for the moments of $h(V_i)$:
\begin{align*}
	\mu_h \equiv \mathbb{E}[h(V_i)], \sigma_h^2 \equiv \mathbb{E}\vert h(V_i) - \mu_h \vert^2, \rho_h \equiv \mathbb{E}\vert h(V_i) - \mu_h \vert^3.
\end{align*}
They are finite when $h$ is a continuous function on $[0,\bar{v}]$.
    \begin{proof} [Proof of Theorem \ref{thm:inid_irreg_welfare}]
        If $h(V_i)$ and $\psi(V_i)$ are perfectly correlated (i.e., $h$ and $\psi$ are linearly dependent), then the result follows from Theorem \ref{thm:hellwig}. Therefore, we only need to study the case where $h(V_i)$ and $\psi(V_i)$ are not perfectly correlated.
        \begin{enumerate}[label = (\roman*)]
            \item 
            We first obtain an asymptotic lower bound for the ex-ante budget:
            \begin{align*}
                \mathbb{E}[b^{n}(\bm{V})] & = \mathbb{E} \left[ \sum \psi(V_i) \mathbf{1}\big\{ \sum h(V_i) \geq \alpha_n \big\} \right] - c_n \mathbb{P}\left( \sum h(V_i) \geq \alpha_n \right) \\
                & \geq \mathbb{E} \left[ \sum \psi(V_i) \mathbf{1}\big\{ \sum h(V_i) \geq \alpha_n \big\} \right] - c_n \\
                & \geq \sqrt{n} \frac{{\sigma}_{\psi h}}{{\sigma}_{h}} \phi\left( \frac{\alpha_n}{\sqrt{n} {\sigma}_{h}} \right) - Cn^{1/4} -c_n,
            \end{align*}
            where the last inequality follows from Lemma \ref{lm:bound-mean-multi}. The first term is strictly positive by Lemma \ref{lm:corr-psi-h}. Then following the same steps as in the proof of Theorem \ref{thm:hellwig}(i), we can show that the leading term in the above expression is $\sqrt{n} \frac{{\sigma}_{\psi h}}{{\sigma}_{h}} \phi\left( \frac{\alpha_n}{\sqrt{n} {\sigma}_{h}} \right)$ under the assumptions $\alpha_n = o(\sqrt{n \log n})$ and $c_n = O_\varepsilon(\sqrt{n})$.
            
            We then derive an upper bound for the ex-ante budget:
            \begin{align*}
                \mathbb{E}[b^{n}(\bm{V})] & \leq \mathbb{E} \left[ \sum \psi(V_i) \mathbf{1}\big\{ \sum h(V_i) \geq \alpha_n \big\} \right] \\
                & \leq \sqrt{n} \frac{{\sigma}_{\psi h}}{{\sigma}_{h}} \phi\left( \frac{\alpha_n}{\sqrt{n} {\sigma}_{h}} \right) + Cn^{1/4} ,
            \end{align*}
            where the last inequality again follows from Lemma \ref{lm:bound-mean-multi}. By the previous analysis, the leading term in the upper bound is also $\sqrt{n} \frac{{\sigma}_{\psi h}}{{\sigma}_{h}} \phi\left( \frac{\alpha_n}{\sqrt{n} {\sigma}_{h}} \right)$, and the result follows.

            \item We follow the proof of Theorem \ref{thm:hellwig}(ii). First notice that $\mathbb{E}[w^{n}(\bm{V})] \leq \mathbb{E}[w^*(\bm{V})]$. 
            We break the ex-ante welfare into two parts
        \begin{align*}
            \mathbb{E}[w^{n}(\bm{V})] & = \mathbb{E} \left[ \sum V_i q^{n} - \sum t^{n}_i(\bm{V}) \right] \\
            & = \mathbb{E} \left[ \left( \sum V_i  -  c_n \right) q^{n}(\bm{V})\right] - \mathbb{E}[b^{n}(\bm{V})].
        \end{align*}
        We obtain a lower bound for the first term on the RHS:
        \begin{align*}
            & \mathbb{E} \left[ \left( \sum V_i  -  c_n \right) q^{n}(\bm{V})\right] \\
            = & \mathbb{E} \left[ \mathbf{1}\big\{ \sum h(V_i) \geq \alpha_n \big\} \sum V_i  \right] - c_n \mathbb{P}\left( \sum h(V_i) \geq c_n \right) \\
            \geq & \mathbb{E} \left[ \mathbf{1}\big\{ \sum h(V_i) \geq \alpha_n \big\} \right] \mathbb{E}\left[ \sum V_i \right] - c_n \mathbb{P}\left( \sum h(V_i) \geq c_n \right)  \\
            = & (n {\mu}- c_n) \mathbb{P}\left( \sum h(V_i) \geq \alpha_n \right) \\
            \geq & (n {\mu}- c_n) \left( 1-\Phi\left( \frac{\alpha_n}{\sqrt{n} {\sigma}_{h}} \right) - \frac{C}{\sqrt{n}} \right)\\
            \geq & n {\mu}- c_n - n \mu\left( \Phi\left( \frac{\alpha_n}{\sqrt{n} {\sigma}_{h}} \right) + \frac{C}{\sqrt{n}} \right),
        \end{align*}
        where the second line follows from the fact that $h$ is non-decreasing and the second to last line follows from Lemma \ref{lm:berry-esseen-uni}. Following the proof of part (i), we can upper bound the ex-ante budget by
        \begin{align*}
            \mathbb{E}[b^{n}(\bm{V})] & \leq \sqrt{n} \frac{{\sigma}_{\psi h}}{{\sigma}_{h}} \phi\left( \frac{\alpha_n}{\sqrt{n} {\sigma}_{h}} \right) + C n^{1/4} \\
            & \leq \sqrt{n} \frac{{\sigma}_{\psi h}}{{\sigma}_{h} \sqrt{2 \pi}} + C n^{1/4}.
        \end{align*}
        We use the same upper bound on the efficient ex-ante welfare $\mathbb{E}[w^*(\bm{V})] \leq n {\mu}$ as before. Combining these results together, we get a similar bound as in part (ii) of Theorem \ref{thm:hellwig}:
    
        \begin{align*}
            \frac{\mathbb{E}[w^*(\bm{V})] - \mathbb{E}[w^{n}(\bm{V})]}{\mathbb{E}[w^*(\bm{V})]} & \leq \frac{c_n}{n{\mu}} + \Phi\left( \frac{\alpha_n}{\sqrt{n} {\sigma}_{h}} \right)  + \frac{C}{\sqrt{n}} \rightarrow 0.
        \end{align*}
        Similar as in the proof of Theorem \ref{thm:hellwig}(ii), the leading term in the above expression is $\Phi\left( \alpha_n / \sqrt{n} {\sigma}_{h} \right)$.
        \item 
        The result comes directly from the previous two parts. 
        \end{enumerate}
    \end{proof}

    \begin{proof} [Proof of Theorem \ref{thm:tightness}]
        It is well-known that the total expected payment of an incentive compatible and individually rational mechanism $(q^n,\{t_i^n\})$ is equal to
        \begin{align*}
            \mathbb{E} \left[ \sum t^n_i(\bm{V}) \right]  = \mathbb{E} \left[ \sum \psi(V_i) q^n(\bm{V}) \right] + \mathbb{E} \left[ \sum t^n_i(0,\bm{V}_{-i}) \right].
        \end{align*}
        See, for example, Equation (3.6) in \cite{hellwig2003public}.
        Since the expected payment $\mathbb{E}[t^n_i(0,\bm{V}_{-i})]$ needs to be non-positive for any $i$ by the individually rational condition, the total expected payment is maximized by setting $q^n(\bm{v}) = \mathbf{1}\{ \sum \psi(v_i) \geq 0 \}$ and $t^n_i(0,\bm{v}_{-i})=0$. Such a mechanism is in fact a special case of the mechanism in Example \ref{eg:vv} with the adjustment term $\alpha_n = 0$. From the proof of Theorem \ref{thm:hellwig}, we can see that the total expected payment is $O(\sqrt{n})$. Notice that even though this particular mechanism may violate incentive compatibility when the distribution is not Myerson regular, the asymptotic order $O(\sqrt{n})$ is nonetheless a valid upper bound on the growth rate of the maximum total expected payment.
    
        Then for any sequence of incentive compatible and individually rational mechanisms $(q^n,\{t^n_i\})$ that satisfies eventually ex-ante budget balanced, it must be true that
        \begin{align*}
            0 \leq \mathbb{E} \left[ \sum t_i(\bm{V}) - c_n q(\bm{V}) \right]  \leq O(\sqrt{n}) - c_n \mathbb{P}(q^n(\bm{V}) = 1), \text{ for $n$ large enough}.
        \end{align*}
        This implies that the provision probability of the public good is converging to zero:
        \begin{align*}
            \mathbb{P}(q^n(\bm{V}) = 1) \leq O(\sqrt{n})/c_n \rightarrow 0. 
        \end{align*}
    \end{proof}
    
    \section{Discussion on Proposition 3 in \cite{hellwig2003public}} \label{sec:comparison-hellwig}

    In this section, we discuss the proof method of Proposition 3 in \cite{hellwig2003public}, hereafter H2003. We first translate the result and the proof with the terminology in our paper. 
    Proposition 3 in H2003 shows that the second-best mechanism is asymptotically efficient when the cost of the public good does not grow with $n$. The proof in H2003 essentially tries to show that the AMT mechanism in Example \ref{eg:vv} is ex-ante budget balanced and asymptotically efficient. Therefore, since the second-best mechanism must have a higher welfare by definition, it is also asymptotically efficient. 
    
    To better explain the proof, we link the notations in H2003 to the ones in our paper. The mechanism $Q^{nk}$ defined in (4.6) and (4.7) on p.597 of H2003 corresponds to the AMT mechanism with $h=\varphi$ in our Example \ref{eg:vv}. The scalar $k$ in the mechanism $Q^{nk}$ corresponds to our adjustment term $\alpha_n$.
    The discussion following Inequality (4.8) on p.598 of H2003 shows that when the adjustment term $k$ is fixed (does not vary with $n$), the mechanism $Q^{nk}$ is ex-ante budget balanced. The discussion following Inequality (4.10) on p.598 of H2003 shows that when the adjustment term $k$ decreases to $-\infty$, the mechanism $Q^{nk}$ is asymptotically efficient. 
    
    We argue that this reasoning is incomplete because a discussion of the ex-ante budget when $k$ varies with $n$ is lacking. More specifically, the proof in H2003 requires that Inequality (4.4) on p.597 to hold for any $\varepsilon >0$ and $n$ sufficiently large. In particular, we can take $\varepsilon = 1/n$. Then the $k(\varepsilon)$ on the second line after (4.10) depends explicitly on $n$ and is decreasing as $n$ increases. In this case, the discussion following Inequality (4.8) is no longer sufficient to show that the mechanism $Q^{nk}$ is ex-ante budget balanced. This is because a decreasing $k$ would decrease the budget. The previous argument that $Q^{nk}$ is ex-ante budget balanced when $k$ is fixed and $n \rightarrow \infty$. However, this does not address the budget when $k$ is decreasing. For example, in the extreme case where $k$ decreases so fast that it is equal to $-\infty$, the revenue becomes zero in the limit, leading to a budget deficit.

    This is where the Berry-Esseen theorem becomes useful: we want to characterize the ex-ante budget when $n$ increases and the adjustment term $k$ decreases. This is possible under the Berry-Esseen theorem because it gives the convergence rate of the central limit theorem. For each $n$ and $k$, we know not only that the budget can be approximated by a normal distribution but also how close this approximation is. Therefore, we can derive the rate of $k$ under which the budget becomes balanced eventually. The details are described in Section \ref{ssec:hellwig} and in the proofs in Appendix \ref{sec:proofs}.

    \end{appendix}


\bibliography{references.bib}
\bibliographystyle{chicago}

\end{document}